\documentclass[11pt,a4paper]{article}
\usepackage[a4paper]{geometry}
\usepackage{amssymb,latexsym,amsmath,amsfonts,amsthm}
\usepackage{graphicx}
\usepackage{epstopdf}
\usepackage{tikz}
\usepackage{pgflibraryshapes}
\usetikzlibrary{arrows,decorations.markings}
\usepackage{subfigure}
\usepackage{overpic}
\usepackage{fullpage}
\usepackage{comment,verbatim}
\usepackage{hyperref}
\usepackage{color}
\usepackage{mathrsfs}
\usepackage[title,titletoc]{appendix}
\usepackage{ulem}

\DeclareMathOperator{\diag}{diag}
\DeclareMathOperator*{\Res}{Res}

\newcommand{\Bes}{\mathrm{Bes}}

\newcommand{\C}{\mathbb{C}}

\newcommand{\Lam}{\Lambda}

\renewcommand{\Re}{\mathrm{Re}\,}
\renewcommand{\Im}{\mathrm{Im}\,}

\newcommand{\Pe}{\mathrm{Pe}}

\renewcommand{\vec}{\mathbf}
\newcommand{\ud}{\,\mathrm{d}}

\newcommand{\ds}{\displaystyle}

\newcommand{\Boh}{\mathcal{O}}

\newtheorem{theorem}{Theorem}[section]

\newtheorem{proposition}[theorem]{Proposition}

\newtheorem{corollary}[theorem]{Corollary}

\newtheorem{rhp}[theorem]{RH problem}

\theoremstyle{definition}

\theoremstyle{remark}

\newtheorem{remark}[theorem]{Remark}

\numberwithin{equation}{section}

\begin{document}

\title{Asymptotics of Fredholm determinant associated with the Pearcey kernel}

\author{Dan Dai\footnotemark[1], ~Shuai-Xia Xu\footnotemark[2] ~and Lun Zhang\footnotemark[3]}

\renewcommand{\thefootnote}{\fnsymbol{footnote}}
\footnotetext[1]{Department of Mathematics, City University of Hong Kong, Tat Chee
Avenue, Kowloon, Hong Kong. E-mail: \texttt{dandai@cityu.edu.hk}}
\footnotetext[2]{Institut Franco-Chinois de l'Energie Nucl\'{e}aire, Sun Yat-sen University,
Guangzhou 510275, China. E-mail: \texttt{xushx3@mail.sysu.edu.cn}}
\footnotetext[3] {School of Mathematical Sciences and Shanghai Key Laboratory for Contemporary Applied Mathematics, Fudan University, Shanghai 200433, China. E-mail: \texttt{lunzhang@fudan.edu.cn }}
\date{\today}

\maketitle

\begin{abstract}
The Pearcey kernel is a classical and universal kernel arising from random matrix theory, which describes the local statistics of eigenvalues when the limiting mean eigenvalue density exhibits a cusp-like singularity. It appears in a variety of statistical physics models beyond matrix models as well. We consider the Fredholm determinant of a trace class operator acting on $L^2\left(-s, s\right)$ with the Pearcey kernel. Based on a steepest descent analysis for a $3\times 3$ matrix-valued Riemann-Hilbert problem, we obtain asymptotics of the Fredholm determinant as $s\to +\infty$, which is also interpreted as large gap asymptotics in the context of random matrix theory.

\end{abstract}

\setcounter{tocdepth}{2} \tableofcontents

\section{Introduction and statement of the results}
In a series of papers \cite{BH}--\cite{BH96a}, Br\'ezin and Hikami initiated the studies of deformed complex Gaussian unitary ensemble (GUE) of the form
\begin{equation}
\frac{1}{Z_n}e^ {-n \textrm{Tr}\left(\frac{M^2}{2}-AM \right)}\ud M,
\end{equation}
defined on the space of $n \times n$ Hermitian matrices, where $Z_n$ is a normalization constant and $A$ is a deterministic matrix also known as the external source. An interesting feature of this matrix ensemble is that it provides a simple model to create a phase transition for the eigenvalues of $M$ in the large $n$ limit. Indeed, by assuming the matrix $A$ is diagonal with two eigenvalues $a$ and $-a$ of equal multiplicity, it follows from the work of Pastur \cite{Pastur} that if $a>1$, the eigenvalues are distributed on two disjoint intervals, while for $0<a<1$, the eigenvalues are distributed on a single interval. In the critical case when $a\to 1$ as $n\to \infty$, the gap closes at the origin and the limiting mean eigenvalue density exhibits a cusp-like singularity, i.e., the density vanishes like $|x|^{1/3}$ as $x\to 0$. Upon letting $n\to\infty$ and after proper scaling, a new local eigenvalue process characterized by the so-called Pearcey kernel emerges near the origin.

The Pearcey kernel $K^\Pe$ is defined as (see \cite{BH,BH1})
\begin{align}\label{eq: pearcey kernel}
K^\Pe(x,y;\rho)&=\int_0^{\infty}p(x+z)q(y+z)\ud z
\nonumber
\\
&=\frac{p(x)q''(y)-p'(x)q'(y)+p''(x)q(y)-\rho
p(x)q(y)}{x-y},
\end{align}
where $\rho\in\mathbb{R}$,
\begin{equation}\label{eq:pearcey integral}
p(x)=\frac{1}{2\pi}\int_{-\infty}^\infty
e^{-\frac14s^4-\frac{\rho}{2}s^2+isx} \ud s \qquad \text{and} \qquad
q(y)=\frac{1}{2\pi} \int_\Sigma e^{\frac14
t^4+\frac{\rho}{2}t^2+ity} \ud t.
\end{equation}
The contour $\Sigma$ in the definition of $q$ consists of the four
rays $\arg t=\pi/4,3\pi/4,5\pi/4,7\pi/4$, where the first and the
third rays are oriented from infinity to zero while the second and
the last rays are oriented outwards; see Figure \ref{fig: sigma} for an illustration. The functions $p$ and $q$ in \eqref{eq:pearcey integral} are solutions of the third order differential equations
\begin{align}
p'''(x)&=xp(x)+\rho p'(x), \label{eq:Pearcey1}
\\
q'''(y)&=-yq(y)+\rho q'(y),
\end{align}
respectively. Since $p$ and $q$ were first introduced by Pearcey in the context of electromagnetic fields \cite{Pear}, the kernel $K^\Pe$ bears the name Pearcey kernel. To see how $K^\Pe$ describe the aforementioned phase transition, note that the eigenvalues of $M$ form a determinantal point process with a correlation kernel $K_n(x,y;a)$ depending on $a$ (see \cite{BH,BH1,Zinn}), it was established in \cite{BH,BH1} (for $\rho=0$) and in \cite{BK3,TW} (for general $\rho \in \mathbb{R}$) that
$$
\lim_{n\to\infty}\frac{1}{n^{3/4}}K_n\left(\frac{x}{n^{3/4}},\frac{y}{n^{3/4}}; 1+\frac{\rho}{2\sqrt{n}}\right)=K^\Pe(x,y;\rho),
$$
i.e., the correlation kernel $K_n$ converges to the Pearcey kernel near the origin as $n\to \infty$ in a double scaling regime.

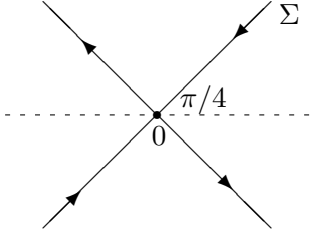
\begin{figure}[h]
\begin{center}
   \setlength{\unitlength}{1truemm}
   \begin{picture}(100,70)(-5,2)
        \dashline{0.8}(20,40)(60,40)
       \put(40,40){\line(1,1){15}}
       \put(40,40){\line(-1,-1){15}}
       \put(40,40){\line(-1,1){15}}
       \put(40,40){\line(1,-1){15}}
       \put(40,40){\thicklines\circle*{1}}
       \put(39.3,36){$0$}
       \put(43,41){$\pi/4$}
       \put(50,50){\thicklines\vector(-1,-1){.0001}}
       \put(30,50){\thicklines\vector(-1,1){.0001}}
       \put(50,30){\thicklines\vector(1,-1){.0001}}
       \put(30,30){\thicklines\vector(1,1){.0001}}
       \put(56,52){$\Sigma$}
\end{picture}
\vspace{-2cm}
\caption{The contour $\Sigma$ in the definition of $q(y)$.}
\label{fig: sigma}
\end{center}
\end{figure}

Like the classical kernels (sine kernel and Airy kernel) arising from random matrix theory \cite{Forrester,metha}, the Pearcey kernel is a universal object as evidenced by its appearance in a variety of stochastic models. On one hand, the Pearcey statistics have been established in specific matrix models including large complex correlated Wishart matrices \cite{HHNa,HHNb}, a two-matrix model with special quartic potential \cite{GZ}, and quite recently for general complex Hermitian Wigner-type matrices at the cusps \cite{EKS}, where the requirement on the identical distribution in Wigner matrices is dropped. It is worthwhile to mention that, for Wigner-type matrices, the density of states exhibits only square root or cubic root cusp singularities; see the classification theorem in \cite{AjEK2017,AlEK2018}. On the other hand, one also encounters the Pearcey kernel beyond matrix models, as can be seen from its connection with non-intersecting Brownian motions at cusps \cite{AOV,AM,BK3} and a combinatorial model on random partitions \cite{OR}.

Let $K^\Pe_{s,\rho}$ be the trace class operator acting on $L^2\left(-s, s\right)$ with the Pearcey kernel \eqref{eq: pearcey kernel}, it is well-known that the associated Fredholm determinant $\det\left(I-K^\Pe_{s,\rho}\right)$ gives us the probability of finding no particles (also known as the gap probability) on the interval $(-s,s)$  in a determinantal point process on the real line characterized by the Pearcey kernel. Moreover, it is shown in \cite{AM,BC1,BH,TW} that the gap probability satisfies some nonlinear differential equations under more general settings. Since one cannot evaluate the Fredholm determinant explicitly for any fixed $s$, a natural and fundamental question is then to ask for its large $s$ asymptotics, which will be the aim of the present work. Denote by
\begin{equation}\label{def:Fnotation}
F(s;\rho):=\ln \det\left(I-K^\Pe_{s,\rho}\right)
\end{equation}
the logarithm of Fredholm determinant associated with the Pearcey kernel, our main result is the following theorem.
\begin{theorem} \label{main-thm}
With $F(s;\rho)$ defined in \eqref{def:Fnotation}, we have, as $s\to +\infty$,
\begin{equation} \label{main-F-asy}
F(s;\rho)= -\frac{9 s^{\frac83}}{2^{\frac{17}3}} + \frac{\rho s^2}{4} - \frac{\rho^2 s^{{\frac43}}}{2^{{\frac{10}3}}} - \frac{2}{9} \ln s +\frac{\rho^4}{216} + C + \Boh(s^{-\frac{2}{3}}),
\end{equation}
uniformly for $\rho$ in any compact subset of $\mathbb{R}$, where $C$ is an undetermined constant independent of $\rho$ and $s$.
\end{theorem}


In the literature, the large $s$ asymptotics of $F(s;\rho)$ was formally derived in \cite{BH} for $\rho = 0$, based on the coupled nonlinear differential equations satisfied by $F(s;0)$. Moreover, the asymptotics therein contains the leading term alone, without providing any information about the error estimate or the sub-leading terms. Our asymptotic expansion  \eqref{main-F-asy} includes more terms and improves the result in \cite{BH}. After a change of variable $s\mapsto s/2$, the leading term of the asymptotic formula \eqref{main-F-asy}, i.e., $-9 s^{\frac83}/2^{\frac{17}3}$, agrees with that obtained in \cite[(3.36)]{BH}. Furthermore, we wish to emphasize that our derivation is rigorous, which makes use of integrable structure of the Pearcey kernel in the sense of Its-Izergin-Korepin-Slavnov \cite{IIKS90} and involves a steepest descent analysis of the relevant Riemann-Hilbert (RH) problem. In a companion work \cite{DXZ2020}, we also establish asymptotics of the deformed Pearcey determinant $\det(I-\gamma K_{s,\rho}^\Pe)$ for $0<\gamma<1$, in the context of thinned Pearcey process, which exhibits significantly different behavior from that of the undeformed case presented in Theorem \ref{main-thm}.


As also observed in \cite{BH}, the leading term of the asymptotic formula \eqref{main-F-asy} confirms the so-called Forrester-Chen-Eriksen-Tracy conjecture \cite{CET95,Forrester93}. This conjecture asserts that if the density of state behaves as $|x-x^*|^\beta$ near a point $x^*$, then the probability $E(s)$ of emptiness of the interval $(x^* -s, x^* +s)$ behaves like
\begin{equation}
  E(s) \sim  \exp \biggl(-C s^{2\beta +2} \biggr), \qquad \textrm{as } s \to +\infty.
\end{equation}
In the present Pearcey case, we have $\beta = \frac{1}{3}$ and $2\beta +2 = \frac83$.

Evaluation of the constant $C$ in \eqref{main-F-asy} is a challenging problem in the studies of large gap asymptotics \cite{Kra1}. For the classical sine, Airy and Bessel kernels encountered in random matrix theory, one could resolve this problem either by investigating the relevant Hankel or Toeplitz determinants which approximate the Fredholm determinants \cite{DIK2008,DIKZ,dkv,Kra2}, or by studying the total integrals of the Painlev\'{e} transcendents on account of Tracy-Widom type formulas for the gap probability \cite{BRD08}; see also \cite{BE,E10,E06} for the approach of operator theory. It seems unlikely that these methods are applicable in the present case. One rough idea to tackle this problem is based on the observation that, as the parameter $\rho$ tends to $-\infty$, the cusp singularity at the origin disappears and the origin becomes a regular point inside the bulk. Thus, we expect that $F(s;\rho)$ might be related to the determinant of (generalized) sine kernel under certain scaling limits when $\rho  \to -\infty$, from which the constant term can be derived. We will leave this issue to a future publication.


Finally, we note that the asymptotics of Fredholm determinant associated with the Pearcey kernel is also investigated from the viewpoint of phase transition in \cite{ACV,BC1}, i.e., to show how the Pearcey process becomes an Airy process by sending both $s$ and the parameter $\rho$ to positive infinity. We emphasize the asymptotic results therein are essentially different from ours.

The rest of this paper is devoted to the proof of Theorem \ref{main-thm}. We mainly follow the general strategy established in \cite{Bor:Dei2002,DIZ97}.  In Section \ref{sec:DiffIdentity}, we relate the partial derivatives of $F(s;\rho)$ to a $3 \times 3$ RH problem with constant jumps, which is essential in the proof. After introducing some auxiliary functions defined on a Riemann surface with a specified sheet structure in Section \ref{sec:auxiliary function}, we then perform a Deift-Zhou steepest descent analysis \cite{DZ93} on this RH problem for large positive $s$ in Section \ref{sec:asymanalyX}. This asymptotic outcome, together with the differential identities for $F(s;\rho)$, will finally lead to the proof of Theorem \ref{main-thm}, as presented in Section \ref{sec:proof}.

\section{Differential identities for the Fredholm determinant}\label{sec:DiffIdentity}
\subsection{A  Riemann-Hilbert characterization of the Pearcey kernel}
The starting point toward the proof of Theorem \ref{main-thm} is an alternative representation of the Pearcey kernel $K^\Pe$ via a $3 \times 3$ RH problem, as shown in \cite{BK3} and stated next.

\begin{rhp} \label{rhp: Pearcey}
We look for a $3 \times 3$ matrix-valued function $\Psi(z)=\Psi(z;\rho)$ satisfying
\begin{itemize}
\item[\rm (1)] $\Psi(z)$ is defined and analytic in $\C \setminus
\{\cup_{j=0}^5\Sigma_j \cup \{ 0 \} \}$,
where
\begin{equation}\label{def:sigmai}
\begin{aligned}
&\Sigma_0=(0,+\infty), ~~ \Sigma_1=e^{\frac{\pi i}{4}}(0,+\infty), ~~\Sigma_2=e^{\frac{ 3 \pi i}{4}}(0,+\infty),
\\
&\Sigma_3=(-\infty, 0), ~~ \Sigma_4=e^{-\frac{3\pi i}{4}}(0,+\infty),~~ \Sigma_5=e^{-\frac{\pi i}{4}}(0,+\infty),
\end{aligned}
\end{equation}
with the orientations as shown in Figure \ref{fig:Pearcey}.

\begin{figure}[h]
\begin{center}
   \setlength{\unitlength}{1truemm}
   \begin{picture}(100,70)(-5,2)
       \put(40,40){\line(-1,-1){20}}
       \put(40,40){\line(-1,1){20}}
       \put(40,40){\line(-1,0){30}}
       \put(40,40){\line(1,0){30}}
       \put(40,40){\line(1,1){20}}
       \put(40,40){\line(1,-1){20}}

       \put(30,50){\thicklines\vector(1,-1){1}}
       \put(30,40){\thicklines\vector(1,0){1}}
       \put(30,30){\thicklines\vector(1,1){1}}
       \put(50,50){\thicklines\vector(1,1){1}}
       \put(50,40){\thicklines\vector(1,0){1}}
       \put(50,30){\thicklines\vector(1,-1){1}}

       \put(39,36.3){$0$}

       \put(20,15){$\Sigma_4$}
       \put(20,63){$\Sigma_2$}
       \put(3,40){$\Sigma_3$}
       \put(60,15){$\Sigma_5$}
       \put(60,63){$\Sigma_1$}
       \put(75,40){$\Sigma_0$}

       \put(22,44){$\Theta_2$}
       \put(22,34){$\Theta_3$}
       \put(55,44){$\Theta_0$}
       \put(55,34){$\Theta_5$}
       \put(38,53){$\Theta_1$}
       \put(38,22){$\Theta_4$}

       \put(40,40){\thicklines\circle*{1}}

   \end{picture}
   \caption{The jump contours $\Sigma_{k}$ and the regions $\Theta_k$, $k=0,1,\ldots,5$, for the RH problem for $\Psi$.}
   \label{fig:Pearcey}
\end{center}
\end{figure}
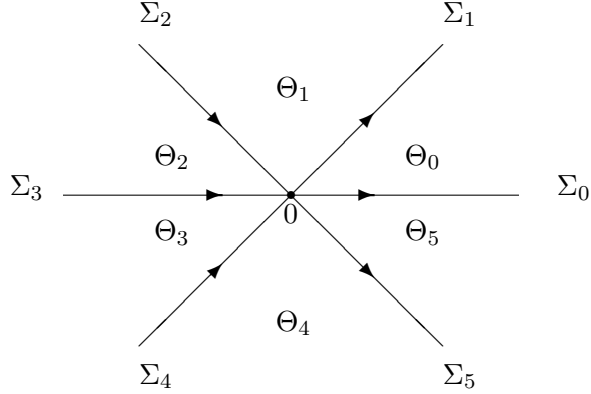

\item[\rm (2)] For $z\in \Sigma_k$, $k=0,1,\ldots,5$, the limiting values
\[ \Psi_+(z) = \lim_{\substack{\zeta \to z \\\zeta\textrm{ on $+$-side of }\Sigma_k}}\Psi(\zeta), \qquad
   \Psi_-(z) = \lim_{\substack{\zeta \to z \\\zeta\textrm{ on $-$-side of }\Sigma_k}}\Psi(\zeta), \]
exist, where the $+$-side and $-$-side of $\Sigma_k$ are the sides
which lie on the left and right of $\Sigma_k$, respectively, when
traversing $\Sigma_k$ according to its orientation. These limiting
values satisfy the jump relation
\begin{equation}\label{jumps:M}
\Psi_{+}(z) = \Psi_{-}(z)J_{\Psi}(z),\qquad z\in \cup_{j=0}^5\Sigma_j,
\end{equation}
where
\begin{equation}\label{def:JPsi}
J_{\Psi}(z):= \left\{
        \begin{array}{ll}
          \begin{pmatrix} 0&1&0 \\ -1&0&0 \\ 0&0&1 \end{pmatrix}, & \qquad \hbox{$z\in \Sigma_0$,} \\
          \begin{pmatrix} 1&0&0 \\ 1&1&1 \\ 0&0&1 \end{pmatrix}, & \qquad \hbox{$z\in \Sigma_1$,} \\
          \begin{pmatrix} 1&0&0 \\ 0&1&0 \\ 1&1&1 \end{pmatrix}, & \qquad \hbox{$z\in \Sigma_2$,} \\
          \begin{pmatrix} 0&0&1 \\ 0&1&0 \\ -1&0&0 \end{pmatrix}, & \qquad \hbox{$z\in \Sigma_3$,} \\
          \begin{pmatrix} 1&0&0 \\ 0&1&0 \\ 1&-1&1 \end{pmatrix}, & \qquad \hbox{$z\in \Sigma_4$,} \\
          \begin{pmatrix} 1&0&0 \\ 1&1&-1 \\ 0&0&1 \end{pmatrix}, & \qquad \hbox{$z\in \Sigma_5$.}
        \end{array}
      \right.
\end{equation}

\item[\rm (3)]  As $z \to \infty$ and $\pm\Im z>0$, we have
\begin{equation}\label{eq:asyPsi}
 \Psi(z)=
\sqrt{\frac{2\pi}{3}}i e^{\frac{\rho^2}{6}} \Psi_0
\left(I+ \frac{\Psi_1}{z} +\mathcal O(z^{-2}) \right)\diag \left(z^{-\frac13},1,z^{\frac13} \right)L_{\pm} e^{\Theta(z)},
\end{equation}
where
\begin{equation} \label{asyPsi:coeff}
\Psi_0 =  \begin{pmatrix}
1 & 0 & 0 \\
0 & 1  & 0 \\
\kappa_3(\rho)+\frac{2\rho}{3} & 0 & 1
\end{pmatrix}, \qquad
\Psi_1 = \begin{pmatrix}
0 & \kappa_3(\rho) & 0 \\
\widetilde\kappa_6(\rho) & 0 & \kappa_3(\rho) + \frac{\rho}{3} \\
0 & \widehat \kappa_6(\rho)  & 0
\end{pmatrix},
\end{equation}
with
\begin{equation} \label{poly:kappa3}
	\kappa_3(\rho) = \frac{\rho^3}{54} - \frac{\rho}{6},
\end{equation}
$\widetilde\kappa_6(\rho) = \kappa_6(\rho) + \frac{\rho}{3} \kappa_3(\rho) - \frac{1}{3}$, $\widehat\kappa_6(\rho) = \kappa_6(\rho) -\kappa_3^2(\rho) + \frac{\rho^2}{9} - \frac{1}{3},$ and
\begin{equation} \label{poly:kappa6}
	\kappa_6(\rho) = \frac{\rho^6}{5832} -  \frac{\rho^4}{162} - \frac{\rho^2}{72} + \frac{7}{36}.
\end{equation}
Moreover, $L_{\pm}$ are constant matrices
\begin{align}\label{def:Lpm}
L_{+}=
\begin{pmatrix}
-\omega & \omega^2 & 1 \\ -1&1&1 \\ -\omega^2 & \omega & 1
\end{pmatrix},
\qquad
L_{-}=
\begin{pmatrix}
\omega^2 & \omega & 1 \\ 1&1&1 \\ \omega & \omega^2 & 1
\end{pmatrix},
\end{align}
with $\omega=e^{2\pi i/3}$, and $\Theta(z)$ is given by
\begin{align}\label{def:Theta}
\Theta(z)=\Theta(z;\rho)&= \begin{cases}
\diag (\theta_1(z;\rho),\theta_2(z;\rho),\theta_3(z;\rho)), & \text{$\Im z >0$,} \\
\diag (\theta_2(z;\rho),\theta_1(z;\rho),\theta_3(z;\rho)), & \text{$\Im z <0$,} \\
\end{cases}
\end{align}
with
\begin{equation} \label{eq: theta-k-def}
\theta_k(z;\rho)=\frac34 \omega^{2k}z^{\frac43}+\frac{\rho}{2}\omega^kz^{\frac23}, \qquad k=1,2,3.
\end{equation}

\item[\rm (4)] $\Psi(z)$ is bounded near the origin.
\end{itemize}
\end{rhp}

It is shown in \cite[Section 8.1]{BK3} that the above RH problem has a unique solution expressed in terms of solutions of the Pearcey differential equation
\eqref{eq:Pearcey1}. Indeed, note that \eqref{eq:Pearcey1} admits the following solutions:
\begin{equation}\label{def:pj}
p_j(z)=p_j(z;\rho)=\int_{\Gamma_j}e^{-\frac14 s^4-\frac{\rho}{2}s^2+is z}\ud s, \qquad j=0,1,\ldots,5,
\end{equation}
where
\begin{align*}
\Gamma_0&=(-\infty,+\infty), \quad && \Gamma_1=(i\infty, 0]\cup[0,\infty),
\\
\Gamma_2&=(i\infty,0]\cup[0,-\infty), \quad && \Gamma_3=(-i\infty, 0]\cup[0,-\infty),
\\
\Gamma_4&=(-i\infty,0]\cup[0,\infty), \quad && \Gamma_5=(-i\infty, i\infty).
\end{align*}
We then have
\begin{equation}\label{def:soltopsi}
\Psi(z)=\left\{
        \begin{array}{ll}
          \begin{pmatrix} -p_2(z) & p_1(z) & p_5(z)\\ -p_2'(z) & p_1'(z) & p_5'(z) \\-p_2''(z) & p_1''(z) & p_5''(z) \end{pmatrix}, & \quad \hbox{$z\in \Theta_0$,} \\
          \begin{pmatrix}  p_0(z) & p_1(z) & p_4(z)\\ p_0'(z) & p_1'(z) & p_4'(z) \\ p_0''(z) & p_1''(z) & p_4''(z) \end{pmatrix}, & \quad \hbox{$z\in \Theta_1$,} \\
          \begin{pmatrix} -p_3(z) & -p_5(z) & p_4(z) \\ -p_3'(z) & -p_5'(z) & p_4'(z) \\-p_3''(z) & -p_5''(z) & p_4''(z) \end{pmatrix}, & \quad \hbox{$z\in \Theta_2$,} \\
          \begin{pmatrix} p_4(z) & -p_5(z) & p_3(z)\\ p_4'(z) & -p_5'(z) & p_3'(z) \\ p_4''(z) & -p_5''(z) & p_3''(z)\end{pmatrix}, & \quad \hbox{$z\in \Theta_3$,} \\
          \begin{pmatrix} p_0(z) & p_2(z) & p_3(z)\\ p_0'(z) & p_2'(z) & p_3'(z) \\ p_0''(z) & p_2''(z) & p_3''(z) \end{pmatrix}, & \quad \hbox{$z\in \Theta_4$,} \\
          \begin{pmatrix} p_1(z) & p_2(z) & p_5(z)\\ p_1'(z) & p_2'(z) & p_5'(z) \\ p_1''(z) & p_2''(z) & p_5''(z) \end{pmatrix}, & \quad \hbox{$z\in \Theta_5$,}
        \end{array}
      \right.
\end{equation}
where $\Theta_k$, $k=0,1,\ldots,5$, is the region bounded by the rays $\Sigma_k$ and $\Sigma_{k+1}$ (with $\Sigma_6:=\Sigma_0$); see Figure \ref{fig: sigma} for an illustration.

\begin{remark} For our purpose, we present a refined asymptotics of $\Psi$ at infinity in \eqref{eq:asyPsi}, which can be verified directly from \eqref{def:soltopsi}. To see this, let us focus on the case $z\in\Theta_1$, since the asymptotics in other regions can be derived in a similar way. Carrying out a steepest descent analysis to the integrals defined \eqref{def:pj} (cf. \cite{BK3,Miyamoto}), we have, as $z\to \infty$,
\renewcommand{\arraystretch}{1.2}
\begin{equation}\label{eq:asyp0}
  p_0(z) = \begin{cases}
     -\sqrt{\frac{2\pi}{3}}i e^{\frac{\rho^2}{6}}  \omega  z^{-\frac{1}{3}} e^{\theta_1(z;\rho)} \biggl( 1 + \frac{\kappa_3(\rho)}{\omega} z^{-\frac{2}{3}} + \frac{\kappa_6(\rho)}{\omega^2} z^{-\frac{4}{3}} + \mathcal O(z^{-2} ) \biggr), & \Im z >0, \\
      \sqrt{\frac{2\pi}{3}}i e^{\frac{\rho^2}{6}}  \omega^2  z^{-\frac{1}{3}} e^{\theta_2(z;\rho)} \biggl( 1 + \frac{\kappa_3(\rho)}{\omega^2} z^{-\frac{2}{3}} + \frac{\kappa_6(\rho)}{\omega^4} z^{-\frac{4}{3}} + \mathcal O(z^{-2} ) \biggr), & \Im z <0,
   \end{cases}
\end{equation}
\begin{equation}
  p_1(z) = \sqrt{\frac{2\pi}{3}}i e^{\frac{\rho^2}{6}}  \omega^2  z^{-\frac{1}{3}} e^{\theta_2(z;\rho)} \biggl( 1 + \frac{\kappa_3(\rho)}{\omega^2} z^{-\frac{2}{3}} + \frac{\kappa_6(\rho)}{\omega^4} z^{-\frac{4}{3}} + \mathcal O(z^{-2} ) \biggr)
\end{equation}
for $-\frac{3\pi}{4} < \arg z < \frac{5\pi}{4},$ and
\begin{equation}\label{eq:asyp4}
  p_4(z) = \sqrt{\frac{2\pi}{3}}i e^{\frac{\rho^2}{6}} z^{-\frac{1}{3}} e^{\theta_3(z;\rho)} \biggl( 1 + \kappa_3(\rho) z^{-\frac{2}{3}} + \kappa_6(\rho) z^{-\frac{4}{3}} + \mathcal O(z^{-2} ) \biggr)
\end{equation}
for $-\frac{\pi}{4} < \arg z < \frac{7\pi}{4},$ where $\theta_k(z;\rho)$, $k=1,2,3$, are given in \eqref{eq: theta-k-def}, $\kappa_3(\rho)$ and $\kappa_6(\rho)$ are polynomials in $\rho$ given in \eqref{poly:kappa3} and \eqref{poly:kappa6}. The asymptotic expansions of $p_j'(z)$ and $p_j''(z)$, $j=0,1,4$, can be derived in a similar fashion or simply by taking derivatives on the right hand sides of \eqref{eq:asyp0}--\eqref{eq:asyp4}. A combination of all these asymptotic results and \eqref{def:soltopsi} then gives us \eqref{eq:asyPsi} after a straightforward calculation.
\end{remark}

Now, define
\begin{equation} \label{eq: tilde-psi}
\widetilde \Psi (z)=\widetilde \Psi (z;\rho)=
\begin{pmatrix}
p_0(z) & p_1(z) & p_4(z)
\\
p_0'(z) & p_1'(z) & p_4'(z)
\\
p_0''(z) & p_1''(z) & p_4''(z)
\end{pmatrix}, \qquad z\in \mathbb{C},
\end{equation}
that is, $\widetilde \Psi$ is the analytic extension of the restriction of $\Psi$ on the region $\Theta_1$ to the whole complex plane. The Pearcey kernel \eqref{eq: pearcey kernel} then admits the following equivalent representation in terms of $\widetilde \Psi$ (see \cite[Equation (10.19)]{BK3}):
\begin{equation}\label{eq:PearceyRH}
K^\Pe(x,y;\rho)=\frac{1}{2 \pi i(x-y)}
\begin{pmatrix}
0 & 1 & 1
\end{pmatrix}
\widetilde \Psi^{-1}(y;\rho)\widetilde \Psi(x;\rho)
\begin{pmatrix}
1
\\
0
\\
0
\end{pmatrix}, \qquad x,y \in \mathbb{R}.
\end{equation}

\subsection{Differential equations for $\Psi$}
For later use, we need the following linear differential equations for $\Psi$ with respect to $z$ and $\rho$.
\begin{proposition}\label{prop:ODE}
Let $\Psi=\Psi(z;\rho)$ be the unique solution of the RH problem \ref{rhp: Pearcey}. We then have
\begin{equation}\label{eq: ODE1}
\frac{\partial \Psi}{\partial z}=
\begin{pmatrix}
0 & 1 & 0
\\
0 & 0 & 1
\\
z & \rho & 0
\end{pmatrix}\Psi,
\end{equation}
and
\begin{equation}\label{eq:ODE2}
\frac{\partial \Psi}{\partial \rho}=\frac12
\begin{pmatrix}
0 & 0 & 1
\\
z & \rho & 0
\\
1 & z & \rho
\end{pmatrix}\Psi.
\end{equation}
\end{proposition}
\begin{proof}
The differential equation \eqref{eq: ODE1} follows directly from \eqref{eq:Pearcey1}, \eqref{def:pj} and \eqref{def:soltopsi}. To show \eqref{eq:ODE2}, we obtain from \eqref{def:pj}, \eqref{eq:Pearcey1} and direct calculations that, for $j=1,\ldots,5$,
\begin{equation}
\frac{\partial p_j}{\partial \rho}=-\frac12 \int_{\Gamma_j}s^2e^{-\frac14 s^4-\frac{\rho}{2}s^2+is z}\ud s=\frac12 p_j'',
\end{equation}
\begin{equation}
\frac{\partial p_j'}{\partial \rho}=-\frac i 2 \int_{\Gamma_j}s^3 e^{-\frac14 s^4-\frac{\rho}{2}s^2+is z}\ud s= \frac12 p_j'''=\frac12 (zp_j+\rho p_j'),
\end{equation}
and
\begin{equation}
\frac{\partial p_j''}{\partial \rho}= \frac 1 2 \int_{\Gamma_j}s^4 e^{-\frac14 s^4-\frac{\rho}{2}s^2+is z}\ud s
= \frac12 p_j''''
=\frac12 (zp_j+\rho p_j')'=\frac12 (p_j+zp_j'+\rho p_j''),
\end{equation}
which leads to \eqref{eq:ODE2}.

This completes the proof of Proposition \ref{prop:ODE}.
\end{proof}

\subsection{Differential identities for $F$}
By \eqref{eq:PearceyRH}, it is readily seen that
\begin{equation}\label{eq:tildeKdef}
K^\Pe(x,y;\rho) = \frac{\vec{f}^t(x)\vec{h}(y)}{x-y},
\end{equation}
where
\begin{equation}\label{def:fh}
\vec{f}(x)=\begin{pmatrix}
f_1
\\
f_2
\\
f_3
\end{pmatrix}:=\widetilde \Psi(x)
\begin{pmatrix}
1
\\
0
\\
0
\end{pmatrix}, \qquad
\vec{h}(y)=\begin{pmatrix}
h_1
\\
h_2
\\
h_3
\end{pmatrix}
:=
\frac{1}{2 \pi i}
\widetilde \Psi^{-t}(y) \begin{pmatrix}
0
\\
1
\\
1
\end{pmatrix}.
\end{equation}

With the function $F$ defined in \eqref{def:Fnotation}, we have
\begin{equation}\label{eq:derivatives}
\frac{\partial}{\partial s}F(s;\rho)=\frac{\ud}{\ud s} \ln \det(I-K_{s,\rho}^\Pe)
=-\textrm{tr}\left((I-K_{s,\rho}^\Pe)^{-1}
\frac{\ud}{\ud s}K_{s,\rho}^\Pe\right)=-R(s,s)-R(-s,-s),
\end{equation}
where $R(u,v)$ stands for the kernel of the resolvent operator, that is,
$$R=\left(I-K_{s,\rho}^\Pe\right)^{-1}-I=K_{s,\rho}^\Pe\left(I-K_{s,\rho}^\Pe\right)^{-1}=\left(I-K_{s,\rho}^\Pe\right)^{-1}K_{s,\rho}^\Pe.$$
Since the kernel of the operator $K_{s,\rho}^\Pe$ is integrable in the sense of \cite{IIKS90}, its resolvent kernel is integrable as well; cf. \cite{DIZ97,IIKS90}. Indeed, by setting
\begin{equation}\label{def:FH}
\vec{F}(u)=
\begin{pmatrix}
F_1 \\
F_2 \\
F_3
\end{pmatrix}:=\left(I-K_{s,\rho}^\Pe\right)^{-1}\vec{f}, \qquad \vec{H}(v)=\begin{pmatrix}
H_1 \\
H_2 \\
H_3
\end{pmatrix}
:=\left(I-K_{s,\rho}^\Pe\right)^{-1}\vec{h},
\end{equation}
we have
\begin{equation}\label{eq:resolventexpli}
R(u,v)=\frac{\vec{F}^t(u)\vec{H}(v)}{u-v}.
\end{equation}

We could also represent $\frac{\partial}{\partial \rho}F(s;\rho)$ in terms of $\vec{F}$ and $\vec{H}$ defined in \eqref{def:FH}. To proceed, we note from \eqref{def:fh} and \eqref{eq:ODE2} that
\begin{equation}\label{def:fderivative}
\frac{\partial \vec{f}}{\partial \rho}(x)
=
\frac12
\begin{pmatrix}
0 & 0 & 1
\\
x & \rho & 0
\\
1 & x & \rho
\end{pmatrix}
\vec{f}(x).
\end{equation}
Moreover, by taking derivative with respect to $\rho$ on both sides of $\Psi \cdot \Psi^{-1}=I$, it is readily seen from \eqref{eq:ODE2} that
\begin{equation}
\frac{\partial \Psi^{-1}}{\partial \rho} = -\frac{\Psi^{-1}}{2}
\begin{pmatrix}
0 & 0 & 1
\\
z & \rho & 0
\\
1 & z & \rho
\end{pmatrix},
\end{equation}
which gives us
\begin{equation}
\frac{\partial \vec{h}}{\partial \rho}(y)
=
-\frac12
\begin{pmatrix}
0 & y & 1
\\
0 & \rho & y
\\
1 & 0 & \rho
\end{pmatrix}
\vec{h}(y).
\end{equation}
This, together with \eqref{eq:tildeKdef} and \eqref{def:fderivative}, implies
\begin{multline}
\frac{\ud}{\ud \rho}K^{\Pe}=\frac{\frac{\partial \vec{f}^t}{\partial \rho}(x)\vec{h}(y)+\vec{f}^t(x)\frac{\partial \vec{h}}{\partial \rho}(y)}{x-y}
\\
=\vec{f}^t(x)
\begin{pmatrix}
0 & 1/2 & 0
\\
0 & 0 & 1/2
\\
0 & 0 & 0
\end{pmatrix}\vec{h}(y)=\frac 12 (f_1(x)h_2(y)+f_2(x)h_3(y)).
\end{multline}
Hence, we obtain
\begin{multline}\label{eq:lamdaderivative}
\frac{\partial}{\partial \rho} F(s;\rho)=
\frac{\ud}{\ud \rho} \ln \det\left(I-K_{s,\rho}^\Pe\right)
=-\textrm{tr}\left(\left(I-K_{s,\rho}^\Pe\right)^{-1}
\frac{\ud}{\ud \rho}K_{s,\rho}^\Pe\right)
\\
=-\frac 12 \int_{-s}^s (F_1(v)h_2(v)+F_2(v)h_3(v))\ud v.
\end{multline}

We next establish the connection between the functions $\frac{\partial}{\partial s} F(s;\rho)$, $\frac{\partial}{\partial \rho} F(s;\rho)$ and an RH problem with constant jumps, which is based on the fact that the resolvent kernel $R(u,v)$ is related to the following RH problem.
\begin{rhp} \label{rhp:Y}
We look for a $3 \times 3$ matrix-valued function
$Y(z)$ satisfying the following properties:
\begin{enumerate}
\item[\rm (1)] $Y(z)$ is defined and analytic in $\mathbb{C}\setminus [-s,s]$, where the orientation is taken from the left to the right.

\item[\rm (2)] For $x\in(-s,s)$, we have
\begin{equation}\label{eq:Y-jump}
 Y_+(x)=Y_-(x)(I-2\pi i \vec{f}(x)\vec{h}^t(x)),
 \end{equation}
where the functions $\vec{f}$ and $\vec{h}$ are defined in \eqref{def:fh}.
\item[\rm (3)] As $z \to \infty$,
\begin{equation}\label{eq:Y-infty}
 Y(z)=I+\frac{\mathsf{Y}_1}{z}+\mathcal O(z^{-2}).
 \end{equation}

\item[\rm (4)] As $z \to \pm s$, we have $Y(z) = \mathcal O(\ln(z \mp s))$.

 \end{enumerate}
\end{rhp}
By \cite{DIZ97}, it follows that
\begin{equation}\label{eq:Yexpli}
Y(z)=I-\int_{-s}^s\frac{\vec{F}(w)\vec{h}^t(w)}{w-z}\ud w
\end{equation}
and
\begin{equation}\label{def:FH2}
\vec{F}(z)=Y(z)\vec{f}(z), \qquad \vec{H}(z)=(Y^t(z))^{-1}\vec{h}(z).
\end{equation}

Recall the RH problem \ref{rhp: Pearcey} for $\Psi$, we make the following undressing transformation to arrive at an RH problem with constant jumps. To proceed, the four rays $\Sigma_k$, $k=1,2,4,5$, emanating from the origin are replaced by their parallel lines emanating from some special points on the real line. More precisely, we replace $\Sigma_1$ and $\Sigma_5$ by their parallel rays $\Sigma_1^{(s)}$ and $\Sigma_5^{(s)}$ emanating from the point $s$, replace $\Sigma_2$ and $\Sigma_4$ by their parallel rays $\Sigma_2^{(s)}$ and $\Sigma_4^{(s)}$ emanating from the point $-s$. Furthermore, these rays, together with the real axis, divide the complex plane into six regions $\texttt{I-VI}$, as illustrated in Figure \ref{fig:X}.

\begin{figure}[h]
\begin{center}
   \setlength{\unitlength}{1truemm}
   \begin{picture}(100,70)(-5,2)
       \put(25,40){\line(-1,0){30}}
       \put(55,40){\line(1,0){30}}

       \dashline{0.8}(25,40)(55,40)

       \put(25,40){\line(-1,-1){25}}
       \put(25,40){\line(-1,1){25}}

       \put(55,40){\line(1,1){25}}
       \put(55,40){\line(1,-1){25}}

       \put(15,40){\thicklines\vector(1,0){1}}
       \put(65,40){\thicklines\vector(1,0){1}}

       \put(10,55){\thicklines\vector(1,-1){1}}
       \put(10,25){\thicklines\vector(1,1){1}}
       \put(70,25){\thicklines\vector(1,-1){1}}
       \put(70,55){\thicklines\vector(1,1){1}}

       \put(39,36.3){$0$}
       \put(-2,11){$\Sigma_4^{(s)}$}

       \put(-2,67){$\Sigma_2^{(s)}$}
       \put(3,42){$\Sigma_3^{(s)}$}
       \put(80,11){$\Sigma_5^{(s)}$}
       \put(80,67){$\Sigma_1^{(s)}$}
       \put(73,42){$\Sigma_0^{(s)}$}

       \put(10,46){$\texttt{III}$}
       \put(10,34){$\texttt{IV}$}
       \put(68,46){$\texttt{I}$}
       \put(68,34){$\texttt{VI}$}
       \put(38,55){$\texttt{II}$}
       \put(38,20){$\texttt{V}$}

       \put(40,40){\thicklines\circle*{1}}
       \put(25,40){\thicklines\circle*{1}}
       \put(55,40){\thicklines\circle*{1}}

       \put(24,36.3){$-s$}
       \put(54,36.3){$s$}

   \end{picture}
   \caption{Regions $\texttt{I-VI}$ and the contours $\Sigma_{k}^{(s)}$, $k=0,1,\ldots,5$, for the RH problem for $X$.}
   \label{fig:X}
\end{center}
\end{figure}
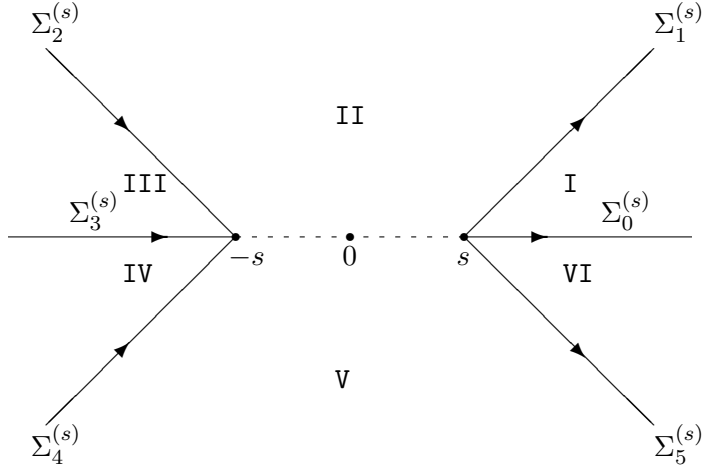

We now define
\begin{align}\label{eq:YtoX}
 X(z) = \left\{
       \begin{array}{ll}
         Y(z)\Psi(z), & \hbox{for $z$ in the region $\texttt{I}\cup \texttt{III}\cup \texttt{IV} \cup \texttt{VI}$,} \\
         Y(z)\widetilde \Psi(z), & \hbox{for $z$ in the region $\texttt{II}$,} \\
         Y(z)\widetilde \Psi(z)
\begin{pmatrix}
1 & -1 & -1
\\
0 & 1 & 0
\\
0 & 0 & 1
\end{pmatrix}, & \hbox{for $z$ in the region $\texttt{V}$,}
       \end{array}
     \right.
\end{align}
where $\widetilde \Psi$ is defined in \eqref{eq: tilde-psi}. Then, $X$ satisfies the following RH problem.

\begin{proposition}\label{rhp:X}
The function $X$ defined in \eqref{eq:YtoX} has the following properties:
\begin{enumerate}
\item[\rm (1)] $X(z)$ is defined and analytic in $\mathbb{C}\setminus \{\cup^5_{j=0}\Sigma_j^{(s)}\cup \{-s\} \cup\{s\}\}$, where
\begin{equation}\label{def:sigmais}
\begin{aligned}
&\Sigma_0^{(s)}=(s,+\infty), ~~ &&\Sigma_1^{(s)}=s+e^{\frac{\pi i}{4}}(0,+\infty), ~~ &&&\Sigma_2^{(s)}=-s+e^{\frac{ 3 \pi i}{4}}(0,+\infty),
\\
&\Sigma_3^{(s)}=(-\infty, -s), ~~ &&\Sigma_4^{(s)}=-s+e^{-\frac{3\pi i}{4}}(0,+\infty), ~~ &&&\Sigma_5^{(s)}=s+e^{-\frac{\pi i}{4}}(0,+\infty),
\end{aligned}
\end{equation}
with the orientations from the left to the right; see the solid lines in Figure \ref{fig:X}.

\item[\rm (2)] $X$ satisfies the jump condition
\begin{equation}\label{eq:X-jump}
 X_+(z)=X_-(z)J_X(z), \qquad z\in \cup^5_{j=0}\Sigma_j^{(s)},
\end{equation}
where
\begin{equation}\label{def:JX}
J_X(z):=\left\{
 \begin{array}{ll}
          \begin{pmatrix} 0&1&0 \\ -1&0&0 \\ 0&0&1 \end{pmatrix}, & \qquad \hbox{$z\in \Sigma_0^{(s)}$,} \\
          \begin{pmatrix} 1&0&0 \\ 1&1&1 \\ 0&0&1 \end{pmatrix},  & \qquad  \hbox{$z\in \Sigma_1^{(s)}$,} \\
          \begin{pmatrix} 1&0&0 \\ 0&1&0 \\ 1&1&1 \end{pmatrix},  & \qquad \hbox{$z\in \Sigma_2^{(s)}$,} \\
          \begin{pmatrix} 0&0&1 \\ 0&1&0 \\ -1&0&0 \end{pmatrix}, & \qquad  \hbox{$z\in \Sigma_3^{(s)}$,} \\
          \begin{pmatrix} 1&0&0 \\ 0&1&0 \\ 1&-1&1 \end{pmatrix}, & \qquad  \hbox{$z\in \Sigma_4^{(s)}$,} \\
          \begin{pmatrix} 1&0&0 \\ 1&1&-1 \\ 0&0&1 \end{pmatrix}, & \qquad  \hbox{$z\in \Sigma_5^{(s)}$.}
        \end{array}
      \right.
 \end{equation}

\item[\rm (3)]As $z \to \infty$ and $\pm \Im z>0$, we have
\begin{equation}\label{eq:asyX}
X(z)=
\sqrt{\frac{2\pi}{3}}i e^{\frac{\rho^2}{6}} \Psi_0
\left(I+ \frac{\mathsf{X}_1}{z} +\mathcal O(z^{-2}) \right)\diag \left(z^{-\frac13},1,z^{\frac13} \right)L_{\pm} e^{\Theta(z)},
\end{equation}
where $\Psi_0$, $L_{\pm}$ and $\Theta(z)$ are given in \eqref{asyPsi:coeff}, \eqref{def:Lpm} and \eqref{def:Theta}, respectively, and
\begin{equation} \label{X1-formula}
	\mathsf{X}_1 = \Psi_1 + \Psi_0^{-1} \mathsf{Y}_1 \Psi_0
\end{equation}
with $\Psi_1$ and $\mathsf{Y}_1$ given in \eqref{asyPsi:coeff} and \eqref{eq:Y-infty}.

\item[\rm (4)] As $z \to \pm s$, we have $X(z) = \mathcal O(\ln(z \mp s))$.
\end{enumerate}
\end{proposition}

\begin{proof}
We only need to show that $X$ does not have a jump over $(-s,s)$, while the other claims follow directly from \eqref{eq:YtoX} and the RH problem \ref{rhp: Pearcey} for $\Psi$.

By \eqref{def:fh}, we have, for $-s<x<s$,
\begin{equation}
I-2\pi i \vec{f}(x)\vec{h}^t(x)=\widetilde\Psi(x)
\begin{pmatrix}
1 & -1 & -1
\\
0 & 1  & 0
\\
0 & 0 & 1
\end{pmatrix}\widetilde\Psi(x)^{-1}.
\end{equation}
This, together with \eqref{eq:Y-jump} and \eqref{eq:YtoX}, implies that for $x\in(-s,s)$,
\begin{multline}
X_+(x)=Y_+(x)\widetilde \Psi(x)=Y_-(x)(I-2\pi i \vec{f}(x)\vec{h}^t(x))\widetilde \Psi(x)
\\
=Y_-(x)\widetilde \Psi(x)\begin{pmatrix}
1 & -1 & -1
\\
0 & 1  & 0
\\
0 & 0 & 1
\end{pmatrix}=X_-(x),
\end{multline}
as desired.

This completes the proof of Proposition \ref{rhp:X}.
\end{proof}

The connections between the above RH problem and the partial derivatives of $F(s;\rho)$ are revealed in the following proposition.
\begin{proposition}\label{prop:derivativeandX}
With $F$ defined in \eqref{def:Fnotation}, we have
\begin{align}
	\frac{\partial}{\partial \rho} F(s;\rho) &=
\frac{\ud}{\ud \rho} \ln \det\left(I-K_{s,\rho}^\Pe\right)
= -\frac{1}{2} \left[ (\mathsf{X}_1)_{12} + (\mathsf{X}_1)_{23} \right] + \frac{\rho^3}{54},  \label{eq:derivativein-rho-X}
\end{align}
where $\mathsf{X}_1$ is given in \eqref{eq:asyX} and $(M)_{ij}$ stands for the $(i,j)$th entry of a matrix $M$, and
\begin{align}
\frac{\partial}{\partial s} F(s;\rho)&=\frac{\ud}{\ud s} \ln \det\left(I-K_{s,\rho}^{\Pe}\right) \nonumber
\\
&=-\frac{1}{ \pi i} \left[\lim_{z \to -s} \left(\left(X^{-1}(z)X'(z)\right)_{21}+\left(X^{-1}(z)X'(z)\right)_{31}\right)
\right],
\label{eq:derivativeinsX-2}
\end{align}
where the above limit is taken from the region $\texttt{II} \, \cup \texttt{V}$.
\end{proposition}
\begin{proof}
We begin with the proof of \eqref{eq:derivativein-rho-X}. From \eqref{eq:Y-infty} and \eqref{eq:Yexpli}, it follows that
\begin{equation}
	\mathsf{Y}_1 = \int_{-s}^s \vec{F}(v)\vec{h}^t(v) \ud v  = \int_{-s}^s \begin{pmatrix}
F_1(v) \\
F_2(v) \\
F_3(v)
\end{pmatrix} \begin{pmatrix}
h_1(v) & h_2(v) & h_3(v)
\end{pmatrix} \ud v .
\end{equation}
This, together with \eqref{eq:lamdaderivative}, implies that
\begin{equation}\label{eq:Frho1}
	\frac{\partial}{\partial \rho} F(s;\rho) = - \frac{1}{2} \left[ (\mathsf{Y}_1)_{12} + (\mathsf{Y}_1)_{23} \right].
\end{equation}
By \eqref{X1-formula}, it is readily seen that $\mathsf{Y}_1 = \Psi_0 (\mathsf{X}_1 - \Psi_1 ) \Psi_0^{-1}$. With the aid of the explicit expressions of $\Psi_0$ and $\Psi_1$ in \eqref{asyPsi:coeff}, we then have
$(\mathsf{Y}_1)_{12} + (\mathsf{Y}_1)_{23}=(\mathsf{X}_1)_{12} + (\mathsf{X}_1)_{23}-\frac{\rho^3}{27}$, which gives us \eqref{eq:derivativein-rho-X} in view of \eqref{eq:Frho1}.

We next consider $\frac{\partial}{\partial s} F(s;\rho)$. For $z\in \texttt{II}$, we see from \eqref{def:fh}, \eqref{def:FH2} and \eqref{eq:YtoX} that
\begin{equation}\label{eq:FHinX}
\vec{F}(z)=Y(z)\vec{f}(z)=Y(z)\widetilde \Psi(z)\begin{pmatrix}
1
\\
0
\\
0
\end{pmatrix}=X(z)
\begin{pmatrix}
1
\\
0
\\
0
\end{pmatrix}
\end{equation}
and
\begin{equation}
\vec{H}(z)=Y^{-t}(z)\vec{h}(z)=X^{-t}(z)\widetilde \Psi^{t}(z)\cdot \frac{\widetilde \Psi^{-t}(z)}{2 \pi i}\begin{pmatrix}
0
\\
1
\\
1
\end{pmatrix}=\frac{X^{-t}(z)}{2 \pi i}
\begin{pmatrix}
0
\\
1
\\
1
\end{pmatrix}.
\end{equation}
Combining the above formulas, \eqref{eq:derivatives}, \eqref{eq:resolventexpli} and L'H\^{o}spital's rule then gives us
\begin{align}
\frac{\partial}{\partial s} F(s;\rho)
=&-\frac{1}{2 \pi i} \bigg[\lim_{z \to -s} \left(\left(X^{-1}(z)X'(z)\right)_{21}+\left(X^{-1}(z)X'(z)\right)_{31}\right)
\nonumber
\\
& +
\lim_{z \to s} \left(\left(X^{-1}(z)X'(z)\right)_{21}+\left(X^{-1}(z)X'(z)\right)_{31}\right)\bigg],
\label{eq:derivativeinsX}
\end{align}
where the above limits are taken from the region $\texttt{II}$. Similarly, one can show that \eqref{eq:derivativeinsX} also holds provided the limits are taken from the region $\texttt{V}$.

The expression \eqref{eq:derivativeinsX} can be further simplified via the following symmetric relation of $X(z)$:
\begin{equation} \label{eq: Xz+-}
  X(z) = \widetilde {C} X(-z) \Lam,
\end{equation}
where
\begin{equation}\label{def:tildeC}
  \widetilde {C}= \Psi_0 \begin{pmatrix}
    -1 & 0 & 0 \\ 0 & 1 & 0 \\ 0 & 0 & -1
  \end{pmatrix} \Psi_0^{-1}
\end{equation}
is a nonsingular matrix with $\Psi_0$ given in \eqref{asyPsi:coeff},
and
\begin{equation} \label{eq: MatLam}
	 \Lam =
\begin{pmatrix}
-1 & 0 & 0
\\
0 & 0 & 1
\\
0 & 1 & 0
\end{pmatrix}.
\end{equation}
To see \eqref{eq: Xz+-}, let us consider the function
\begin{equation}\label{def:tildeX}
  \widetilde X(z):=X(-z) \Lam.
\end{equation}
It is straightforward to check that $\widetilde X(z)$ satisfies the same jump condition as $X(z)$ shown in \eqref{eq:X-jump} and \eqref{def:JX}. This means $X(z)\widetilde X^{-1}(z)$ is analytic in the complex plane with possible isolated singular points located at $z = \pm s$. Since $ X(z) = \mathcal O(\ln(z \mp s)) $  as $z \to \pm s$, we conclude that the possible singular points $\pm s$ are removable. In view of the asymptotics of $X(z)$ given in \eqref{eq:asyX}, we further obtain from \eqref{def:tildeX} and a bit more cumbersome calculation that
\begin{equation*}
  X(z)\widetilde X^{-1}(z) = \widetilde{C} + \Boh(z^{-1}), \qquad \textrm{as } z \to \infty,
\end{equation*}
with $\widetilde{C}$ given in \eqref{def:tildeC}. An appeal to Liouville's theorem then leads us to \eqref{eq: Xz+-}.

As a consequence of \eqref{eq: Xz+-}, we have
$$X'(z) = -\widetilde {C} X'(-z) \Lam. $$
This, together with \eqref{eq: Xz+-} and the fact that $\Lam^{-1}=\Lam$, implies
\begin{eqnarray*}
  \lim_{z \to s} \left( X^{-1}(z)X'(z) \right) =  - \lim_{z \to s} \left( \Lam X^{-1}(-z) X'(-z) \Lam \right)  =  - \lim_{z \to -s} \left( \Lam X^{-1}(z) X'(z) \Lam \right) .
\end{eqnarray*}
To this end, we observe that for an arbitrary $3\times 3$ matrix $ M = (m_{ij})_{i,j=1}^3$,
\begin{equation*}
 \left( \Lam M \Lam  \right)_{21} = - m_{31} \quad \textrm{and} \quad  \left( \Lam M \Lam  \right)_{31} = - m_{21}.
\end{equation*}
A combination of the above two formulas shows that
$$\lim_{z \to -s} \left(X^{-1}(z)X'(z)\right)_{21}+\left(X^{-1}(z)X'(z)\right)_{31}
=
\lim_{z \to s} \left(X^{-1}(z)X'(z)\right)_{21}+\left(X^{-1}(z)X'(z)\right)_{31}.$$
Inserting this formula into \eqref{eq:derivativeinsX} gives us \eqref{eq:derivativeinsX-2}.

This completes the proof of Proposition \ref{prop:derivativeandX}.
\end{proof}

\section{Auxiliary functions} \label{sec:auxifuncs}
\label{sec:auxiliary function}

In this section, we introduce some auxiliary functions and study
their properties. The aim is to construct the so-called
$\lambda$-functions, of which the analytic continuation defines a
meromorphic function on a Riemann surface with a specified sheet
structure. The $\lambda$-functions have desired behavior around each
branch point, and will be crucial in our further asymptotic analysis of the RH problem \ref{rhp:X} for $X$.

Throughout this section, unless specified differently, we shall take
the principal branch for all fractional powers.

\subsection{A three-sheeted Riemann surface and the $w$-functions}
\label{sec:w function}
We introduce a three-sheeted Riemann surface $\mathcal R$ with sheets
\begin{align*}
\mathcal R_1 &= \C \setminus \{(-\infty,-1]\cup[1,+\infty) \}, & \mathcal R_2 &= \C \setminus [1,+\infty), \\
\mathcal R_3 &=\C \setminus (-\infty,-1] .
\end{align*}
We connect the sheets $\mathcal R_j$, $j=1,2,3$, to each other in the usual crosswise manner along the cuts $(-\infty,-1]$ and $[1,+\infty)$. More precisely, $\mathcal R_1$ is connected to $\mathcal R_2$ along the cut $[1,+\infty)$ and $\mathcal R_1$ is connected to $\mathcal R_3$ along the cut $(-\infty,-1]$. We then compactify the resulting surface by adding a common point at $\infty$ to the sheets $\mathcal R_1$ and $\mathcal R_2$, and a common point at $\infty$ to the sheets $\mathcal R_1$ and $\mathcal R_3$. We denote this compact Riemann surface by $\mathcal R$, which has genus zero and is shown in Figure \ref{fig: Riemann surface}.

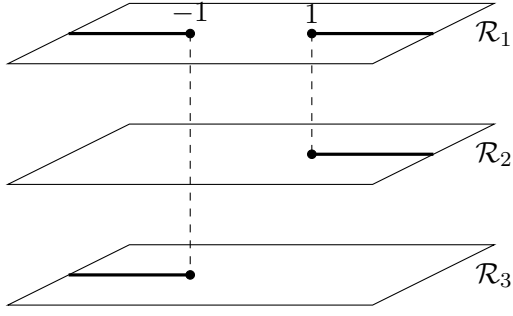
\begin{figure}[h]
\centering
\begin{tikzpicture}[scale=0.8]
\draw (0,0)--(6,0)--(8,1)--(2,1)--cycle
      (0,-2)--(6,-2)--(8,-1)--(2,-1)--cycle
      (0,-4)--(6,-4)--(8,-3)--(2,-3)--cycle;
\draw[very thick] (1,0.5)--(3,0.5) (5,0.5)--(7,0.5)
                  (1,-3.5)--(3,-3.5)
                  (5,-1.5)--(7,-1.5);
\draw (8,0.5) node{$\mathcal R_1$}
      (8,-1.5) node{$\mathcal R_2$}
      (8,-3.5) node{$\mathcal R_3$};
\draw[dashed] (3,0.5)--(3,-3.5);
\draw[dashed] (5,0.5)--(5,-1.5);
\filldraw [black] (3,0.5) circle (2pt) node [above] (q1) {$-1$};
\filldraw [black] (5,0.5) circle (2pt) node [above] (q1) {$1$};
\filldraw [black] (5,-1.5) circle (2pt) node [above] (q1) {};
\filldraw [black] (3,-3.5) circle (2pt) node [above] (q1) {};
\end{tikzpicture}
\caption{Riemann surface $\mathcal R$.}
\label{fig: Riemann surface}
\end{figure}

We intend to find functions $\lambda_j$, $j=1,2,3$, on these sheets, such that each $\lambda_j$ is analytic on $\mathcal R_j$ and admits an analytic continuation across the cuts. For this purpose, we start with an elementary function $w(z)$ that is
meromorphic on $\mathcal R$, which satisfies the following algebraic equation
\begin{equation} \label{eq: alg eq w}
w(z)^3-3w(z)+2z=0.
\end{equation}
We choose three solutions $w_j(z)$, $j = 1,2,3$, to \eqref{eq: alg eq w} such that they are defined and analytic on $\mathcal R_j$, respectively. Each function $w_j(z)$ maps $\mathcal R_j$ to certain domain in the extended complex $w$-plane $\overline \C$. The correspondences between some points $z\in \mathcal R$ and the points $w \in \overline \C$ are given in Table \ref{Table:z-w-corres}, where $z^{(j)}$ denotes the point $z$ on the closure of the $j$-th sheet $\mathcal R_j$.
\begin{table}[h]
  \centering
  \begin{tabular}{|c||l|l|l|l|l|l|}
   \hline  & & & &  & & \\ [-1em]
   $z \in \mathcal R$ & $-1^{(1)}$ & $\infty^{(1)}$ & $1^{(1)}$ & $0^{(1)}$ & $0^{(2)}$ & $0^{(3)}$ \\ [1ex]
    \hline & & & &  & & \\ [-1em]  $w \in \overline \C$ & $-1$ & $\infty$ & 1 & 0 & $\sqrt{3}$ & $-\sqrt{3}$ \\ [1ex]
   \hline
  \end{tabular}
\medskip
  \caption{Images of some points $z\in \mathcal R$ under the mappings $w_j$, $j=1,2,3$.} \label{Table:z-w-corres}
\end{table}
Due to the sheet structure shown in Figure \ref{fig: Riemann surface}, it follows that $\infty^{(1)} = \infty^{(2)} = \infty^{(3)}$, $-1^{(1)} = -1^{(3)}$ and $1^{(1)} = 1^{(2)}$. Moreover, we actually have the following explicit expressions for $w_j(z)$.

\begin{proposition} \label{prop:w-expression}
The three solutions $w_j(z)$, $j=1,2,3$, to the algebraic equation \eqref{eq: alg eq w} with the mapping properties shown in Table \ref{Table:z-w-corres} are given by
\begin{equation} \label{eq: w-eta}
  w_j(z) = \omega^{j-2}\eta (z)^{\frac{1}{3}} + \omega^{2-j}\eta (z)^{-\frac{1}{3}},
\end{equation}
where
\begin{equation} \label{eq: eta-def}
  \eta(z) = (z^2 -1 )^{\frac{1}{2}} - z, \qquad z \in \mathbb{C} \setminus \{ (-\infty, -1] \cup [1,+\infty)\},
\end{equation}
with
\begin{equation} \label{eq: eta-arg}
  \arg \eta(z) \in (0, \pi).
\end{equation}
\end{proposition}

\begin{proof}
Clearly, $\eta(z)\neq 0$ and satisfies the quadratic equation
  \begin{equation}\label{eq:etaeq}
    \eta(z)^2 + 2z \eta(z) + 1 = 0.
  \end{equation}
Hence, with $w_j(z)$, $j=1,2,3$, defined in \eqref{eq: w-eta}, we have
  \begin{align*}
    w_j(z)^3 -3w_j(z) + 2z & =  \left( \omega^{j-2}\eta (z)^{\frac{1}{3}} + \omega^{2-j}\eta (z)^{-\frac{1}{3}} \right)^3 - 3 \left( \omega^{j-2}\eta (z)^{\frac{1}{3}} + \omega^{2-j}\eta (z)^{-\frac{1}{3}} \right) + 2z
     \\
     & =  \eta(z) + \frac{1}{\eta(z)} + 2z = 0,
  \end{align*}
as expected, where in the last step we have made use of \eqref{eq:etaeq}. Furthermore, an elementary analysis shows that $\Im \eta(z) >0$ for $z \in \mathbb{C} \setminus \{ (-\infty, -1] \cup [1,+\infty)\}$. Thus, $\arg \eta(z) \in (0, \pi)$ and $\eta(0) = e^{\frac{\pi i}{2}}=i $. This, together with
\eqref{eq: w-eta}, implies that $w_1(0) = 0$,  $ w_2(0) = \sqrt{3}$, and $w_3(0) = -\sqrt{3}$. The other relations in Table \ref{Table:z-w-corres}
can then be verified similarly, and we omit the details here.

This completes the proof of Proposition \ref{prop:w-expression}.
\end{proof}

As a consequence of Proposition \ref{prop:w-expression}, we have the following properties of $w_j(z)$.
\begin{proposition} \label{prop:w}
The functions $w_j(z)$, $j=1,2,3$, given in \eqref{eq: w-eta} satisfy the following
properties.
\begin{itemize}
\item[\rm (a)] $w_j(z)$ is analytic on $\mathcal R_j$, $j=1,2,3$, and
\begin{align}
w_{1,\pm}(x) &= w_{3,\mp}(x), \qquad x \in (-\infty,-1), \label{eq: w1pm-w3pm}
\\
w_{1,\pm}(x) &= w_{2,\mp}(x), \qquad x \in (1,\infty). \label{eq: w1pm-w2pm}
\end{align}
Here, we orient $(-\infty,-1)$ and $(1,\infty)$ from the left to the right. Hence, the function
$\cup_{j=1}^3 \mathcal R_j \to \C: \mathcal R_j \ni z \mapsto
w_j(z)$ has an analytic continuation to a meromorphic function $w:
\mathcal R \to \overline \C$. This function is a bijection.

\item[\rm (b)] $w$ satisfies the symmetry properties
\begin{align}
w_j(\overline z) &= \overline{w_j(z)},   \qquad~~\,~~\,  z\in\mathcal R_j, \quad j=1,2,3, \label{eq: wz-wzbar} \\
w_1(-z) & = - w_1(z) ,   \qquad~~\,  z\in  \C\setminus \{ (-\infty,-1] \cup [1,+\infty) \} , \label{eq: w1z+-} \\
w_2(-z)&=-w_3(z),   \qquad~~\,  z \in \C\setminus (-\infty,-1]. \label{eq: w2-w3}
\end{align}

\item[\rm (c)] As $z \to \infty$ and $-\pi<\arg z <\pi$, we have
\begin{equation} \label{eq: w1-large-z}
w_1(z) =\left\{
          \begin{array}{ll}
            -2^{\frac13}\omega^2z^{\frac13}-\frac{\omega}{2^{\frac13}} z^{-\frac13}+\frac{\omega^2}{6\cdot 2^{\frac23}}z^{-\frac53}-\frac{\omega}{12\cdot 2^{\frac13}}z^{-\frac73}+\frac{\omega^2}{18\cdot 4^{\frac13}}z^{-\frac{11}{3}}+\Boh(z^{-\frac{13}{3}}), & \hbox{$\Im z>0$,}
\\
          [2mm] -2^{\frac13}\omega z^{\frac13}-\frac{\omega^2}{2^{\frac13}} z^{-\frac13}+\frac{\omega}{6\cdot 2^{\frac23}}z^{-\frac53}-\frac{\omega^2}{12\cdot 2^{\frac13}}z^{-\frac73}+\frac{\omega}{18\cdot 4^{\frac13}}z^{-\frac{11}{3}}+\Boh(z^{-\frac{13}{3}}), & \hbox{$\Im z<0$,}
          \end{array}
        \right.
\end{equation}
and
\begin{multline} \label{eq: w3-large-z}
w_3(z) = -2^{\frac13} z^{\frac13}-\frac{1}{2^{\frac13}} z^{-\frac13}+\frac{1}{6\cdot 2^{\frac23}}z^{-\frac53}-\frac{1}{12\cdot 2^{\frac13}}z^{-\frac73}
\\
+\frac{1}{18\cdot 4^{\frac13}}z^{-\frac{11}{3}}+\Boh(z^{-\frac{13}{3}}), \quad z\in\C \setminus (-\infty,-1].
\end{multline}

\item[\rm (d)] As $z \to -1$ and $-\pi<\arg(z+1) <\pi$, we have
\begin{equation} \label{eq: w1-n1-asy}
w_1(z)=-1+\sqrt{\frac23}(z+1)^{\frac12}+\frac{z+1}{9}+\frac{5}{54\cdot \sqrt{6}}(z+1)^{\frac32}+\Boh(z+1)^2,
\end{equation}
and
\begin{equation} \label{eq: w3-n1-asy}
w_3(z)=-1-\sqrt{\frac23}(z+1)^{\frac12}+\frac{z+1}{9}-\frac{5}{54\cdot \sqrt{6}}(z+1)^{\frac32}+\Boh(z+1)^2.
\end{equation}

\item[\rm (e)] As $z \to 1$ and $-\pi<\arg(z-1) <\pi$, we have
\begin{equation} \label{eq: w1-p1-asy}
w_{1}(z)= \left\{
            \begin{array}{ll}
              1+i\sqrt{\frac23}(z-1)^{\frac12}+\frac{z-1}{9}-i\frac{5}{54\cdot \sqrt{6}}(z-1)^{\frac32}+\Boh(z-1)^2, & \hbox{$\Im z>0$,} \\
            [2mm]  1-i\sqrt{\frac23}(z-1)^{\frac12}+\frac{z-1}{9}+i\frac{5}{54\cdot \sqrt{6}}(z-1)^{\frac32}+\Boh(z-1)^2, & \hbox{$\Im z<0$.}
            \end{array}
          \right.
\end{equation}
and
\begin{equation} \label{eq: w2-p1-asy}
w_{2}(z)= \left\{
            \begin{array}{ll}
              1-i\sqrt{\frac23}(z-1)^{\frac12}+\frac{z-1}{9}+i\frac{5}{54\cdot \sqrt{6}}(z-1)^{\frac32}+\Boh(z-1)^2, & \hbox{$\Im z>0$,} \\
            [2mm]  1+i\sqrt{\frac23}(z-1)^{\frac12}+\frac{z-1}{9}-i\frac{5}{54\cdot \sqrt{6}}(z-1)^{\frac32}+\Boh(z-1)^2, & \hbox{$\Im z<0$.}
            \end{array}
          \right.
\end{equation}

\end{itemize}
\end{proposition}
\begin{proof}
To show \eqref{eq: w1pm-w3pm}, we see from \eqref{eq: w-eta} that, if $x<-1$,
\begin{align*}
w_{1,+}(x)&=\omega^{-1}\eta_{+}(x)^{\frac13}+\omega\eta_{+}(x)^{-\frac13}=\omega^{-1}(-x-\sqrt{x^2-1})^{\frac13}+\omega(-x-\sqrt{x^2-1})^{-\frac13}
\nonumber
\\
&=\omega^{-1}\eta_{-}(x)^{-\frac13}+\omega\eta_{-}(x)^{\frac13}=w_{3,-}(x).
\end{align*}
Similarly, it is easy to check $w_{1,-}(x)=w_{3,+}(x)$ for $x<-1$ and \eqref{eq: w1pm-w2pm}.

While \eqref{eq: wz-wzbar} follows directly from  \eqref{eq: w-eta}, the proof of  \eqref{eq: w1z+-} relies on the fact that
\begin{equation*}
  \eta(z) \eta(-z) = -1.
\end{equation*}
Hence, by \eqref{eq: w-eta}, it follows that
\begin{align*}
  w_1(-z) &= \omega^{-1} \eta(-z)^{\frac{1}{3}} + \omega  \eta(-z)^{-\frac{1}{3}} = \omega^{-1} e^{\frac{1}{3} \pi i}  \eta(z)^{-\frac{1}{3}} + \omega e^{-\frac{1}{3} \pi i}   \eta(z)^{\frac{1}{3}} \\
  & = - \omega^{-1}  \eta(z) ^{\frac{1}{3}} -\omega \eta(z) ^{-\frac{1}{3}}  = -w_1(z),
\end{align*}
which is \eqref{eq: w1z+-}. The relation \eqref{eq: w2-w3} follows in a similar manner.

To obtain the asymptotics of $w_j(z)$, $j=1,3$, as $z \to \infty$, we observe from the definition of $\eta(z)$ in \eqref{eq: eta-def} that
\begin{eqnarray*}
  \eta(z) = \begin{cases}
    \ds -\frac{1}{2z} - \frac{1}{8z^3} - \frac{1}{16z^5} + \Boh(z^{-7}), & \Im z >0, \vspace{1mm} \\
    \ds -2z +\frac{1}{2z} + \frac{1}{8z^3} + \frac{1}{16z^5} + \Boh(z^{-7}), & \Im z <0.
  \end{cases}
\end{eqnarray*}
Inserting the above formula into \eqref{eq: w-eta}, it is readily seen that, if $\Im z>0$ and $z \to \infty$,
\begin{align*}
  w_1(z) & = \omega^{-1} \eta(z) ^{\frac{1}{3}} + \omega \eta(z)^{-\frac{1}{3}} \\
  & =  \frac{e^{-\frac{\pi i}{3}}}{\sqrt[3]{2} z^{\frac{1}{3}}} \left(1 +   \frac{1}{4z^2} + \frac{1}{8z^4} + \Boh(z^{-6}) \right)^{\frac{1}{3}} + e^{\frac{\pi i}{3}} \sqrt[3]{2} z^{\frac{1}{3}} \left(1 +   \frac{1}{4z^2} + \frac{1}{8z^4} + \Boh(z^{-6}) \right)^{-\frac{1}{3}}
  \\
  &=  -2^{\frac13}\omega^2z^{\frac13}-\frac{\omega}{2^{\frac13}} z^{-\frac13}+\frac{\omega^2}{6\cdot 2^{\frac23}}z^{-\frac53}-\frac{\omega}{12\cdot 2^{\frac13}}z^{-\frac73}+\frac{\omega^2}{18\cdot 4^{\frac13}}z^{-\frac{11}{3}}+\Boh(z^{-\frac{13}{3}}),
\end{align*}
which is the first formula in \eqref{eq: w1-large-z}. This, together with the fact that $w_1(-z)=-w_1(z)$, implies the second formula of \eqref{eq: w1-large-z}. The asymptotics of $w_3(z)$ in \eqref{eq: w3-large-z} can be derived through similar computations, we omit the details here.

We next come to the asymptotics of $w_j(z)$, $j=1,3$, as $z \to -1$. Since
$$(z-1)^{\frac12}=\sqrt{2}i-\frac{i}{2\sqrt{2}}(z+1)+\Boh(z+1)^2, \qquad z\in\mathbb{C}\setminus [1,+\infty),\qquad z\to -1,$$
it follows from \eqref{eq: eta-def} that
\begin{align*}
  \eta(z) = (z^2-1)^{\frac12}-z = 1 + i \sqrt{2} (z+1)^{\frac{1}{2}} - (z+1) - i \frac{ (z + 1)^{\frac{3}{2}}}{ 2 \sqrt2}  + \Boh(z+1)^{\frac52}, \qquad z\to -1,
\end{align*}
with $-\pi <\arg(z+1)< \pi$. A combination of the above formula and \eqref{eq: w-eta} then gives us
\begin{align*}
  w_1(z) & = \omega^{-1} \eta(z)^{\frac{1}{3}} + \omega \eta(z)^{-\frac{1}{3}}
  \\
   & =  \omega^{-1} \left(1 + i \sqrt{2} (z+1)^{\frac{1}{2}} - (z+1) - i \frac{ (z + 1)^{\frac{3}{2}}}{ 2 \sqrt2}
    + \Boh(z+1)^{\frac52} \right)^{\frac{1}{3}}
    \\
    & \quad + \omega \left(1 + i \sqrt{2}  (z+1)^{\frac{1}{2}} - (z+1) - i \frac{ (z + 1)^{\frac{3}{2}}}{ 2 \sqrt2}  + \Boh(z+1)^{\frac52} \right)^{-\frac{1}{3}}
    \\
    &=-1+\sqrt{\frac23}(z+1)^{\frac12}+\frac{z+1}{9}+\frac{5}{54\cdot \sqrt{6}}(z+1)^{\frac32}+\Boh(z+1)^2,
\end{align*}
as shown in \eqref{eq: w1-n1-asy}. The asymptotics of $w_3$ in \eqref{eq: w1-n1-asy} can be proved in a similar manner, we omit the details here.

Finally, we note that, as $z \to 1$,
\begin{eqnarray*}
  \eta(z) = \left\{
              \begin{array}{ll}
                -1  +  \sqrt{2}(z-1)^{\frac{1}{2}} - (z-1) + \frac{(z - 1)^{\frac{3}{2}}}{ 2 \sqrt2}  + \Boh(z-1)^{\frac52}, & \hbox{$\Im z>0$,} \\
                 -1  -  \sqrt{2}(z-1)^{\frac{1}{2}} - (z-1) - \frac{(z - 1)^{\frac{3}{2}}}{ 2 \sqrt2}  + \Boh(z-1)^{\frac52}, & \hbox{$\Im z<0$,}
              \end{array}
            \right.
\end{eqnarray*}
with $-\pi <\arg(z-1)< \pi$. Inserting the above formula into \eqref{eq: w-eta} then gives us \eqref{eq: w1-p1-asy} and \eqref{eq: w2-p1-asy} after straightforward calculations.

This completes the proof of Proposition \ref{prop:w}.
\end{proof}

The image of the map $w:\mathcal R\mapsto \overline{\mathbb{C}}$ is illustrated in Figure \ref{fig: image of w}.

\begin{figure}[h]
\centering
\begin{tikzpicture}[scale=0.6]
\draw (-7.5,0)--(7.5,0) ;
\draw[very thick] (-4,0).. controls (-4,2) and  (-5,4)..(-7,5)
      (4,0).. controls (4,2) and  (5,4)..(7,5);
\draw[very thick] (-4,0).. controls (-4,-2) and  (-5,-4)..(-7,-5)
           (4,0).. controls (4,-2) and  (5,-4)..(7,-5);
\draw[dashed] (-7,5.5)--(7,-5.5) (-7,-5.5)--(7,5.5);

\draw (5mm,0mm) arc(0:45:5mm);
\draw (1.3,-0.2) node[above]{$\pi/3$};

\filldraw    (0,0) circle (2pt)  (-4,0) circle (2pt) (4,0) circle (2pt);
\draw  (0,0) node[below right]{$0$}
      (-4,0) node[below right]{$-1$}
      (4,0) node[below right]{$1$};
\draw (-6.1,-3)node[right]{$\gamma_1^-$}
      (-6.1, 3)node[right]{$\gamma_1^+$}
      (5,-3)node[right]{$\gamma_2^-$}
      (5,3)node[right]{$\gamma_2^+$}
      (-0.5,4.5)node[right]{$ \widehat{\mathcal R}_1$}
      (-6,1)node[right]{$ \widehat{\mathcal R}_3$}
       (6,1)node[right]{$ \widehat{\mathcal R}_2$};

\end{tikzpicture}
\caption{Image of the map $w:\mathcal R\mapsto
\overline{\mathbb{C}}$. The thick lines $\gamma_i^{\pm}$, $i=1,2$, are the images of
the cuts in the Riemann surface $\mathcal R$ under this map. More precisely, $\gamma_1^\pm=w_{1,\pm}((-\infty,-1))$, $\gamma_2^\pm=w_{1,\pm}((1,\infty))$ and $w(\mathcal R_k)=\widehat{\mathcal R}_k$, $k=1,2,3$. }
\label{fig: image of w}
\end{figure}
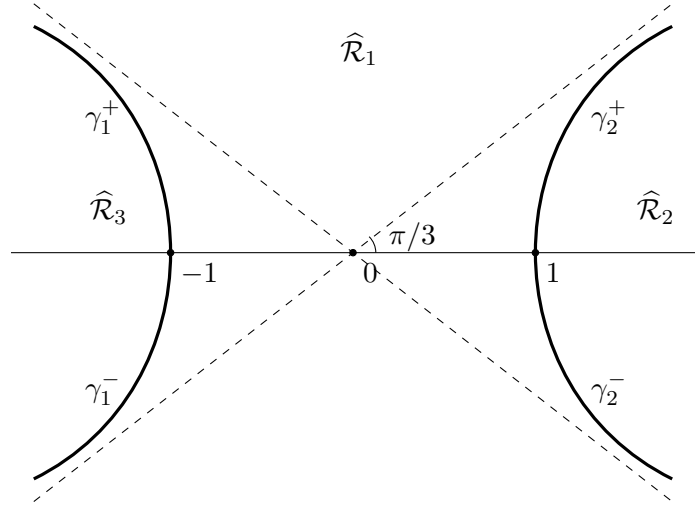

\subsection{The $\lambda$-functions}
\label{sec:lamda function}

With the functions $w_j(z)$ given in Proposition \ref{prop:w-expression}, we define the
$\lambda$-functions as
\begin{equation}\label{def: lambda function}
\lambda_j(z)=\frac{3}{4^{\frac53}}w_j(z)^4+\left(\frac{\rho}{2^{\frac53}s^{\frac23}}-\frac{3}{2^{\frac43}}\right)w_j(z)^2, \qquad j=1,2,3,
\end{equation}
which depend on the parameters $s>0$ and $\rho \in \mathbb{R}$.
The properties of the $\lambda$-functions are listed in the following proposition.
\begin{proposition}\label{prop: prop of lambda}
The functions $\lambda_j(z)$, $j=1,2,3$, defined by \eqref{def: lambda function} have the
following properties.
\begin{itemize}
\item[\rm (a)] $\lambda_j(z)$ is analytic on $\mathcal R_j$, $j=1,2,3$, and
\begin{align}
\lambda_{1,\pm}(x) &= \lambda_{3,\mp}(x),  \qquad x \in (-\infty,-1), \label{eq:lamda 1 lamda 3}
\\
\lambda_{1,\pm}(x) &= \lambda_{2,\mp}(x), \qquad x \in (1,\infty),
\end{align}
Hence the function $\cup_{j=1}^3 \mathcal R_j \to \C:  \mathcal
R_j \ni z \mapsto \lambda_j(z)$ has an analytic continuation to a
meromorphic function on the Riemann surface $\mathcal R $.

\item[\rm (b)] We have the following symmetry properties
\begin{align}
\lambda_j(\overline z) & = \overline{\lambda_j(z)}, \qquad z \in \mathcal R_j, \quad j=1,2,3, \label{eq:symlambdaj}
\\
\lambda_1(-z) &  = \lambda_1(z) ,   \qquad  z\in  \C\setminus \{ (-\infty,-1] \cup [1,+\infty) \} , \label{eq: lambda1z+-}
\\
\lambda_2(-z)&=\lambda_3(z), \qquad z \in \C\setminus (-\infty,-1]. \label{eq: lambda23z}
\end{align}

\item[\rm (c)]
As $z \to \infty$ and $-\pi<\arg z <\pi$, we have
\begin{equation}
\lambda_1(z)=\left\{
               \begin{array}{ll}
                 \frac34 \omega^2 z^{\frac43}+\frac{\rho \omega}{2s^{\frac23}}z^{\frac23}-D_0+D_1 \omega^2 z^{-\frac23}+\Boh(z^{-\frac43}), & \hbox{$\Im z>0$,} \\ [1ex]
                 \frac34 \omega z^{\frac43}+\frac{\rho \omega^2}{2s^{\frac23}}z^{\frac23}-D_0+D_1 \omega z^{-\frac23} +\Boh(z^{-\frac43}), & \hbox{$\Im z<0$,}
               \end{array}
             \right.
\end{equation}
and
\begin{equation}
\lambda_3(z)= \frac34 z^{\frac43}+\frac{\rho}{2s^{\frac23}}z^{\frac23}-D_0+D_1 z^{-\frac23} + \Boh(z^{-\frac43}), \quad z\in\C \setminus (-\infty,-1],
\end{equation}
where
\begin{equation} \label{D0-def}
D_0:=3\cdot 2^{-\frac73}-(2s)^{-\frac23}\rho, \qquad D_1:=\frac{1}{8}\left(-2+\left(\frac2s\right)^{\frac23}\rho\right).
\end{equation}

\item[\rm (d)] As $z\to -1$  and $-\pi<\arg (z+1) <\pi$, we have
\begin{equation}\label{eq:lambda-1 at -1}
\lambda_1(z)=C_0+C_1(z+1)^{\frac12}+C_2(z+1)+C_3(z+1)^{\frac32}+\Boh(z+1)^2,
\end{equation}
and
\begin{equation}\label{eq:lambda-3 at -1}
\lambda_3(z)=C_0-C_1(z+1)^{\frac12}+C_2(z+1)-C_3(z+1)^{\frac32}+\Boh(z+1)^2,
\end{equation}
where
\begin{equation}\label{def:Ci}
\begin{aligned}
C_0&:=\frac{\rho}{2^{\frac 53}s^{\frac 23}}-\frac{9}{2^{\frac {10}{3}}},  \quad &&C_1:=\frac{\sqrt{3}}{2^{\frac56}}-\frac{\rho}{\sqrt{3}\cdot2^{\frac16}s^{\frac23}},
\\
C_2&:=\frac{2^{\frac23}}{3}+\frac{2^{\frac13}\rho}{9s^{\frac23}}, \quad &&C_3:=\frac{7\rho}{108\sqrt{3}\cdot 2^{\frac16}s^{\frac23}}-\frac{165}{108\sqrt{3}\cdot2^{\frac56}}.
\end{aligned}
\end{equation}

\item[\rm (e)] As $z\to 1$  and $-\pi<\arg (z-1) <\pi$, we have
\begin{equation}
\lambda_1(z)=\left\{
               \begin{array}{ll}
C_0-iC_1(z-1)^{\frac12}-C_2(z-1)+iC_3(z-1)^{\frac32}+\Boh(z-1)^2, & \hbox{$\Im z>0$,}
 \\
C_0+iC_1(z-1)^{\frac12}-C_2(z-1)-iC_3(z-1)^{\frac32}+\Boh(z-1)^2, & \hbox{$\Im z<0$,}
               \end{array}
             \right.
\end{equation}
and
\begin{equation}
\lambda_2(z)=\left\{
               \begin{array}{ll}
                 C_0+iC_1(z-1)^{\frac12}-C_2(z-1)-iC_3(z-1)^{\frac32}+\Boh(z-1)^2, & \hbox{$\Im z>0$,} \\
                 C_0-iC_1(z-1)^{\frac12}-C_2(z-1)+iC_3(z-1)^{\frac32}+\Boh(z-1)^2, & \hbox{$\Im z<0$,}
               \end{array}
             \right.
\end{equation}
where the constants $C_i$, $i=0,1,2,3$, are given in \eqref{def:Ci}.

\item[\rm (f)] For $z\in \C\setminus \{ (-\infty,-1] \cup [1,+\infty) \}$, we have
\begin{equation}\label{eq:sumlambda}
    \lambda_1(z) + \lambda_2(z) + \lambda_3(z) = - \frac{9}{2^{\frac{7}{3}}} + \frac{3 \rho}{2^{\frac{2}{3}} s^{\frac{2}{3}}}.
  \end{equation}
\end{itemize}
\end{proposition}
\begin{proof}
The proofs of items (a)--(e) follow directly from the definition of $\lambda_j(z)$ in \eqref{def: lambda function} and Proposition \ref{prop:w}. It then remains to prove \eqref{eq:sumlambda}.
Since $w_j$, $j=1,2,3$, are three solutions of the algebraic equation \eqref{eq: alg eq w}, it follows from Vieta's rule that, for $z\in \C\setminus \{ (-\infty,-1] \cup [1,+\infty) \}$,
\begin{align*}
w_1(z)+w_2(z)+w_3(z)&=0,
\\
w_1(z)w_2(z)+w_1(z)w_3(z)+w_2(z)w_3(z)&=-3,
\\
w_1(z)w_2(z)w_3(z)&=-2z.
\end{align*}
Hence,
\begin{align}\label{eq:sumwsquare}
&w_1(z)^2+w_2(z)^2+w_3(z)^2
\nonumber
\\
&=(w_1(z)+w_2(z)+w_3(z))^2-2(w_1(z)w_2(z)+w_1(z)w_3(z)+w_2(z)w_3(z))=6,
\end{align}
\begin{align*}
&w_1(z)^2w_2(z)^2+w_1(z)^2w_3(z)^2+w_2(z)^2w_3(z)^2
\\
&=(w_1(z)w_2(z)+w_1(z)w_3(z)+w_2(z)w_3(z))^2-2w_1(z)w_2(z)w_3(z)(w_1(z)+w_2(z)+w_3(z))
\\
&=9,
\end{align*}
and
\begin{align}\label{eq:sumwquartic}
&w_1(z)^4+w_2(z)^4+w_3(z)^4
\nonumber \\
&=(w_1(z)^2+w_2(z)^2+w_3(z)^2)^2-2(w_1(z)^2w_2(z)^2+w_1(z)^2w_3(z)^2+w_2(z)^2w_3(z)^2)
\nonumber \\
&=36-18=18.
\end{align}
A combination of \eqref{def: lambda function}, \eqref{eq:sumwsquare} and \eqref{eq:sumwquartic} yields \eqref{eq:sumlambda}.

This completes the proof of Proposition \ref{prop: prop of lambda}.
\end{proof}

%

\section{Asymptotic analysis of the Riemann-Hilbert problem for $X$}\label{sec:asymanalyX}
In this section, we shall perform a Deift-Zhou steepest descent analysis \cite{DZ93} to the RH problem for $X$ as $s \to +\infty$. It consists of a series of explicit and invertible transformations which leads to an RH problem tending to the identity matrix as $s \to +\infty$.

\subsection{First transformation: $X \to T$}
This transformation is a rescaling of the RH problem for $X$, which is  defined by
\begin{equation}\label{def:XtoT}
T(z)=X(sz).
\end{equation}
It is then straightforward to check that $T$ satisfies the following RH problem.

\begin{rhp}
The function $T$ defined in \eqref{def:XtoT} has the following properties:
\begin{enumerate}
\item[\rm (1)] $T(z)$ is defined and analytic in $\mathbb{C}\setminus \{\cup^5_{j=0}\Sigma_j^{(1)}\cup \{-1\} \cup\{1\}\}$, where the contours $\Sigma_j^{(1)}$ are defined in \eqref{def:sigmais} with $s=1$.

\item[\rm (2)] $T$ satisfies the jump condition
\begin{equation}\label{eq:T-jump}
 T_+(z)=T_-(z)J_T(z), \qquad z\in\cup^5_{j=0}\Sigma_j^{(1)},
\end{equation}
where
\begin{equation}\label{def:JT}
J_T(z):=\left\{
 \begin{array}{ll}
          \begin{pmatrix} 0&1&0 \\ -1&0&0 \\ 0&0&1 \end{pmatrix}, &\qquad  \hbox{$z\in \Sigma_0^{(1)}$,} \\
          \begin{pmatrix} 1&0&0 \\ 1&1&1 \\ 0&0&1 \end{pmatrix},  &\qquad  \hbox{$z\in \Sigma_1^{(1)}$,} \\
          \begin{pmatrix} 1&0&0 \\ 0&1&0 \\ 1&1&1 \end{pmatrix},  &\qquad  \hbox{$z\in \Sigma_2^{(1)}$,} \\
          \begin{pmatrix} 0&0&1 \\ 0&1&0 \\ -1&0&0 \end{pmatrix}, &\qquad  \hbox{$z\in \Sigma_3^{(1)}$,} \\
          \begin{pmatrix} 1&0&0 \\ 0&1&0 \\ 1&-1&1 \end{pmatrix}, &\qquad  \hbox{$z\in \Sigma_4^{(1)}$,} \\
          \begin{pmatrix} 1&0&0 \\ 1&1&-1 \\ 0&0&1 \end{pmatrix}, &\qquad  \hbox{$z\in \Sigma_5^{(1)}$.}
        \end{array}
      \right.
 \end{equation}

\item[\rm (3)]As $z \to \infty$ and $\pm\Im z>0$, we have
\begin{equation}\label{eq:asyT}
 T(z)=
\sqrt{\frac{2\pi}{3}}i e^{\frac{\rho^2}{6}} \Psi_0
\left(I+\frac{\mathsf{X}_1}{sz} + \mathcal O(z^{-2}) \right)\diag \left((sz)^{-\frac13},1,(sz)^{\frac13} \right)L_{\pm}e^{\Theta(sz)},
\end{equation}
where we recall that
\begin{align}
\Theta(sz)&= \begin{cases}
\diag (\theta_1(sz;\rho),\theta_2(sz;\rho),\theta_3(sz;\rho)), & \text{$\Im z >0$,} \\
\diag (\theta_2(sz;\rho),\theta_1(sz;\rho),\theta_3(sz;\rho)), & \text{$\Im z <0$,} \\
\end{cases}
\end{align}
with
\begin{equation}\label{eq:thetasz}
\theta_k(sz;\rho)=s^{\frac43}\left(\frac34 \omega^{2k}z^{\frac43}+\frac{\rho\omega^k}{2s^{\frac23}}z^{\frac23}\right), \qquad k=1,2,3.
\end{equation}

\item[\rm (4)] As $z \to \pm 1$, we have $T(z) = \mathcal O(\ln(z \mp 1))$.
\end{enumerate}
\end{rhp}

\subsection{Second transformation: $T \to S$}
In this transformation we normalize the large $z$ behavior of $T$ using the $\lambda$-functions introduced in Section \ref{sec:lamda function}. We define
\begin{equation}\label{def:TtoS}
S(z)= S_0 \frac{\diag \left(s^{\frac13},1,s^{-\frac13} \right)}{e^{D_0s^{\frac43}} \sqrt{\frac{2\pi}{3}}i e^{\frac{\rho^2}{6}}} \Psi_0^{-1}  T(z)\diag\left(e^{-s^{\frac43}\lambda_1(z)},e^{-s^{\frac43}\lambda_2(z)},e^{-s^{\frac43}\lambda_3(z)}\right),
\end{equation}
where
\begin{equation}
  S_0 =  \begin{pmatrix}
1 & 0 & 0 \\
0 & 1  & 0 \\
D_1 s^{\frac{4}{3}} & 0 & 1
\end{pmatrix}
\end{equation}
with $D_0$ and $D_1$ being the constants given in \eqref{D0-def}, and the functions $\lambda_i$, $i=1,2,3$, are defined in \eqref{def: lambda function}. Then, $S$ satisfies the following RH problem.

\begin{proposition}\label{prop:S}
The function $S$ defined in \eqref{def:TtoS} has the following properties:
\begin{enumerate}
\item[\rm (1)] $S(z)$ is defined and analytic in $\mathbb{C}\setminus \{\cup^5_{j=0}\Sigma_j^{(1)}\cup \{-1\} \cup\{1\}\}$.

\item[\rm (2)] $S$ satisfies the jump condition
\begin{equation}\label{eq:S-jump}
 S_+(z)=S_-(z)J_S(z), \qquad z\in \cup^5_{j=0}\Sigma_j^{(1)},
\end{equation}
where
\begin{equation}\label{def:JS}
J_S(z):=\left\{
 \begin{array}{ll}
          \begin{pmatrix} 0&1&0 \\ -1&0&0 \\ 0&0&1 \end{pmatrix}, &\quad  \hbox{$z\in \Sigma_0^{(1)}$,} \\
          \begin{pmatrix} 1&0&0 \\
         e^{s^{\frac43}(\lambda_2(z)-\lambda_1(z))}&1&e^{s^{\frac43}(\lambda_2(z)-\lambda_3(z))}
         \\
          0&0&1
        \end{pmatrix},  &\quad  \hbox{$z\in \Sigma_1^{(1)}$,} \\
          \begin{pmatrix} 1&0&0 \\ 0&1&0 \\ e^{s^{\frac43}(\lambda_3(z)-\lambda_1(z))} & e^{s^{\frac43}(\lambda_3(z)-\lambda_2(z))} &1 \end{pmatrix},  &\quad  \hbox{$z\in \Sigma_2^{(1)}$,} \\
          \begin{pmatrix} 0&0&1 \\ 0&1&0 \\ -1&0&0 \end{pmatrix}, &\quad  \hbox{$z\in \Sigma_3^{(1)}$,} \\
          \begin{pmatrix} 1&0&0 \\ 0&1&0 \\ e^{s^{\frac43}(\lambda_3(z)-\lambda_1(z))} &-e^{s^{\frac43}(\lambda_3(z)-\lambda_2(z))}&1 \end{pmatrix}, &\quad  \hbox{$z\in \Sigma_4^{(1)}$,} \\
          \begin{pmatrix} 1&0&0 \\ e^{s^{\frac43}(\lambda_2(z)-\lambda_1(z))}&1&-e^{s^{\frac43}(\lambda_2(z)-\lambda_3(z))} \\ 0&0&1 \end{pmatrix}, &\quad  \hbox{$z\in \Sigma_5^{(1)}$.}
        \end{array}
      \right.
 \end{equation}

\item[\rm (3)]As $z \to \infty$ and $\pm\Im z>0$, we have
\begin{equation}\label{eq:asyS}
 S(z)=
\left(I + \frac{\mathsf{S}_1}{z} +  \mathcal O(z^{-2}) \right) \diag \left(z^{-\frac13},1,z^{\frac13} \right)L_{\pm},
\end{equation}
where
\begin{equation} \label{asyS:coeff}
\mathsf{S}_1 = \begin{pmatrix}
* & \frac{(\mathsf{X}_1)_{12}}{ s^{\frac{2}{3}}} - D_1 s^{\frac{4}{3}} & * \\
* & *  & \frac{(\mathsf{X}_1)_{23}}{ s^{\frac{2}{3}}} - D_1 s^{\frac{4}{3}} \\
* & * & *
\end{pmatrix}
\end{equation}
with the constant $D_1$ given in \eqref{D0-def}, and $*$ stands for some unimportant entry.

\item[\rm (4)] As $z \to \pm 1$, we have $S(z) = \mathcal O(\ln(z \mp 1))$.

\end{enumerate}
\end{proposition}
\begin{proof}
By \eqref{def:TtoS}, it is clear that $S(z)$ is defined and analytic in $\mathbb{C}\setminus \{\cup^5_{j=0}\Sigma_j^{(1)}\cup \{-1\} \cup\{1\}\}$ and
\begin{align}
J_S(z)&=S_-^{-1}(z)S_+(z)
\nonumber \\
&=\diag\left(e^{s^{\frac43}\lambda_{1,-}(z)},e^{s^{\frac43}\lambda_{2,-}(z)},e^{s^{\frac43}\lambda_{3,-}(z)}\right)J_T(z)
\nonumber \\
& \quad \times \diag\left(e^{-s^{\frac43}\lambda_{1,+}(z)},e^{-s^{\frac43}\lambda_{2,+}(z)},e^{-s^{\frac43}\lambda_{3,+}(z)}\right),
\end{align}
for $z\in \cup^5_{j=0}\Sigma_j^{(1)}$, where $J_T$ is given in \eqref{def:JT}. This, together with item (a) in Proposition \ref{prop: prop of lambda}, gives us \eqref{def:JS}.

Checking the asymptotics \eqref{eq:asyS} is a bit more cumbersome. To that end, we observe from Proposition \ref{prop: prop of lambda} and \eqref{eq:thetasz} that the following formula is useful
\begin{multline*}
\diag\left(e^{-s^{\frac43}\lambda_1(z)+\theta_1(sz)},e^{-s^{\frac43}\lambda_2(z)+\theta_2(sz)},e^{-s^{\frac43}\lambda_3(z)+\theta_3(sz)}\right)
\\
=e^{D_0s^{\frac43}}\left(I- \frac{D_1s^{\frac43}}{z^{\frac23}}\diag\left(\omega^2,\omega,1\right)+ \Boh(z^{-\frac43})\right), \qquad z\to \infty.
\end{multline*}
We omit the details here.

This completes the proof of Proposition \ref{prop:S}.
\end{proof}

\subsection{Estimate of $J_S$ for large $s$}
The jump matrix $J_S(z)$ defined in \eqref{def:JS} is constant on $\Sigma_0^{(1)}$ and on $\Sigma_3^{(1)}$. On the other jump contours, it comes out that the nonzero off-diagonal entries of $J_S$ are all exponentially small for large $s$, which can be seen from the following proposition. In what follows, we denote by $U(z_0,\delta)$ the fixed open disk centered at $z_0$ with small radius $\delta>0$.

\begin{proposition}\label{prop:lambda}
\begin{enumerate}
\item[\rm (a)] There exist positive constants $c_1,c_2>0$ such that
\begin{align}
    \Re (\lambda_2(z)-\lambda_1(z) ) &\leq - c_1 |z|^{\frac43},&&  z\in (\Sigma_1^{(1)} \cup \Sigma_5^{(1)} )\setminus U(1,\delta), \label{eq:lambda21}\\
    \Re (\lambda_3(z)-\lambda_1(z) ) &\leq - c_2 |z|^{\frac43},&&  z\in (\Sigma_2^{(1)} \cup \Sigma_4^{(1)} )\setminus U(-1,\delta),
\end{align}
for $s$ large enough.
\item[\rm (b)]
There exists a constant $c_3>0$ such that
\begin{align}
    \Re \left(\lambda_2(z)-\lambda_3(z) \right) &\leq - c_3 |z|^{\frac43},&&  z\in \Sigma_1^{(1)} \cup \Sigma_5^{(1)}, \label{eq:lambda23ineq}
\\
    \Re \left(\lambda_3(z)-\lambda_2(z) \right) &\leq - c_3 |z|^{\frac43},&&  z\in \Sigma_2^{(1)} \cup \Sigma_4^{(1)},
\label{eq:lambda32ineq}
 \end{align}
for $s$ large enough.
\end{enumerate}
\end{proposition}
\begin{proof}
We will only present the proof of \eqref{eq:lambda21}, since the proofs of other estimates are analogous. In view of the symmetry properties \eqref{eq:symlambdaj}, it suffices to show \eqref{eq:lambda21} for $z\in \Sigma_1^{(1)}\setminus U(1,\delta)$. To this end, let us define
\begin{equation}
\lambda_j^*(z)=\frac{3}{4^{\frac53}}w_j(z)^4-\frac{3}{2^{\frac43}}w_j(z)^2, \qquad j=1,2,3.
\end{equation}
From the behavior of the $\lambda$-functions at infinity given in Proposition \ref{prop: prop of lambda}, it follows from an elementary analysis that
\begin{equation}
|\lambda_j(z)-\lambda_j^*(z)|\leq \varrho s^{-\frac23}\max\{1,|z|^\frac43\}, \qquad z\in \Sigma_1^{(1)}\setminus U(1,\delta),
\end{equation}
for some $\varrho>0$, if $s$ is large enough. This, together with the triangle inequality, implies that it is sufficient to show \eqref{eq:lambda21} for $\lambda_j^*(z)$, $j=1,2$.

For large value of $z\in\Sigma_1^{(1)}$, the estimate follows from the asymptotics of $\lambda_j^*$ at infinity, which can be obtained by taking $s\to +\infty$ in item (c) of Proposition \ref{prop: prop of lambda}. For bounded $z$, the claim is supported by Figure \ref{fig: estimates 1}. The sign of $\Re \left( \lambda_2^*- \lambda_1^*\right)$ remains unchanged in the region bounded by the solid lines in Figure \ref{fig: estimates 1}. By using the asymptotics of $\lambda_1^*$ and $\lambda_2^*$ at infinity,  it is readily seen that $\Re \left( \lambda_2^*(z)- \lambda_1^*(z)\right)<0$ for all $z$ on the right side of the solid curve excluding $[1, \infty)$, which particularly holds for the finite part of $\Sigma_1^{(1)}\setminus U(1,\delta)$.

This completes the proof of Proposition \ref{prop:lambda}.
\end{proof}

\begin{figure}
\hspace{4cm}
\includegraphics[width=340pt]{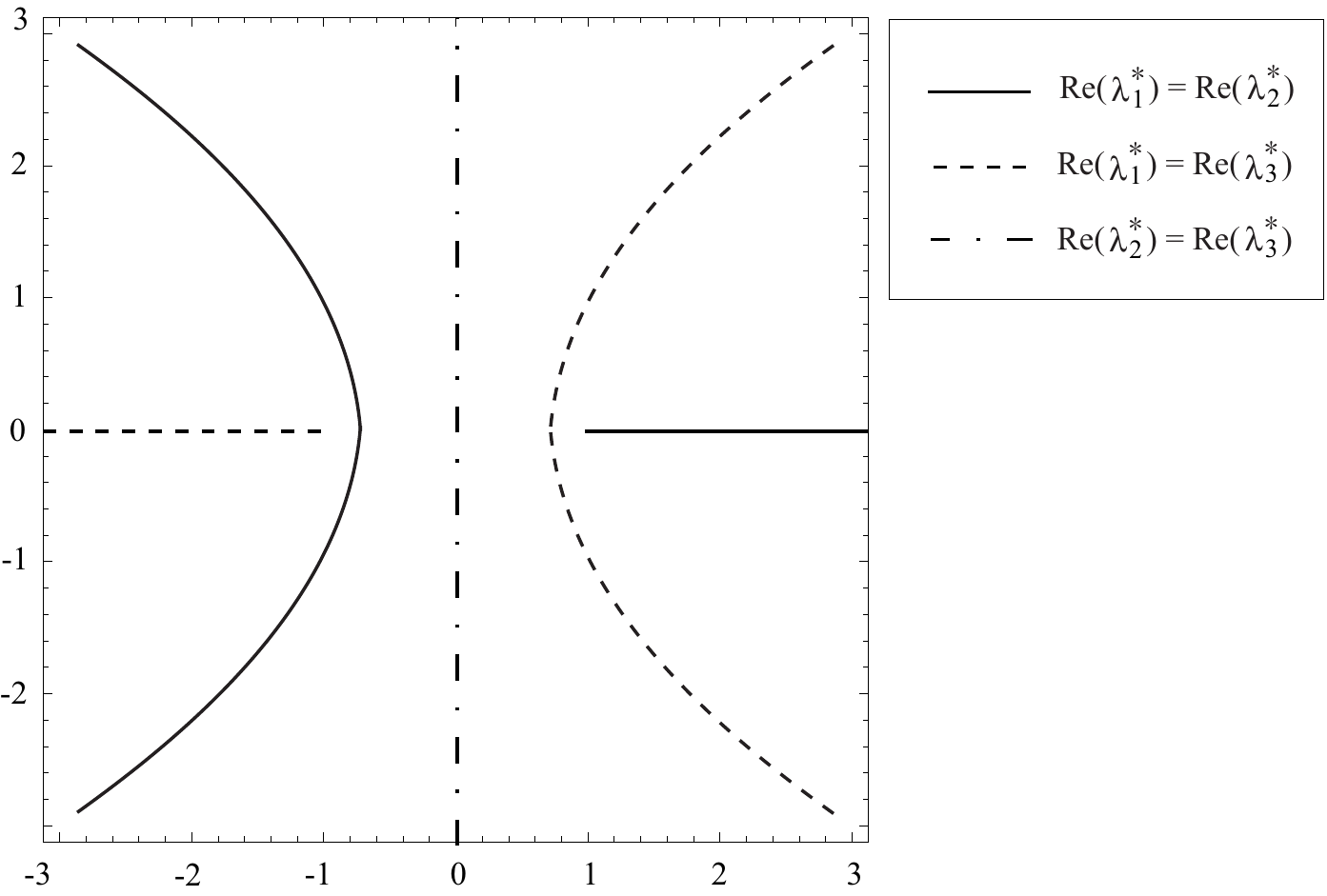}
\caption{Curves where $\Re \lambda_1^*=\Re \lambda_2^*$ (solid lines), $\Re \lambda_1^*=\Re \lambda_3^*$ (dashed lines), and $\Re \lambda_2^*=\Re \lambda_3^*$ (dashed-dotted lines).} 
\label{fig: estimates 1}
\end{figure}

As a consequence of the above proposition, the following corollary about the estimate of $J_S$ is immediate.
\begin{corollary} \label{coro:JS}
    There is a constant $c > 0$ such that
    \[ J_{S}(z) = I + \Boh \left(e^{-c s^{\frac43} |z|^{\frac43}}\right), \qquad \text{ as } s \to +\infty, \]
     uniformly for $z \in
\bigcup\limits_{i=1,2,4,5}\Sigma_i^{(1)}\setminus\{U(1,\delta) \cup U(-1,\delta)\}$.
\end{corollary}

\subsection{Global parametrix}
By Corollary \ref{coro:JS}, if we suppress all entries of the jump matrices for $S$ that decay exponentially as $s \to +\infty$, we are led to the following
RH problem for the global parametrix $N$.
\begin{rhp}\label{rhp:global para}
We look for a $3 \times 3$ matrix-valued function $N$ satisfying
\begin{itemize}
\item[\rm (1)]
$N(z)$ is defined and analytic in $\mathbb{C}\setminus \{(-\infty,-1]\cup[1,+\infty)\}$.
\item[\rm (2)] $N$ satisfies the jump condition
\begin{equation}\label{eq:N-jump}
 N_+(x)=N_-(x)\left\{
                \begin{array}{ll}
                  \begin{pmatrix} 0&1&0 \\ -1&0&0 \\ 0&0&1 \end{pmatrix}, & \qquad \hbox{$x>1$,}  \\
                  \begin{pmatrix} 0&0&1 \\ 0&1&0 \\ -1&0&0 \end{pmatrix}, & \qquad \hbox{$x<-1$.}
                \end{array}
              \right.
\end{equation}

\item[\rm (3)]
As $z \to \infty$, we have
\begin{equation}\label{eq:asyN1}
N(z)=
 \left( I+  \Boh (z^{-1}) \right) \diag \left( z^{-\frac13},1,z^{\frac13} \right) L_{\pm},
\end{equation}
where the constants $L_{\pm}$ are given in \eqref{def:Lpm}.
\end{itemize}
\end{rhp}

The above RH problem can be solved explicitly.  Let the scalar functions $N_j(\cdot)$, $j = 1,2,3,$ be defined as
\begin{align}
N_1(w)=\frac{1-\frac13 w^2}{(1-w^2)^{\frac12}}, \quad N_2(w)=-\frac{i}{\sqrt{6}}\frac{w-\frac{\sqrt{3}}{3}w^2}{(1-w^2)^{\frac12}}, \quad
N_3(w)=-\frac{i}{\sqrt{6}}\frac{w+\frac{\sqrt{3}}{3}w^2}{(1-w^2)^{\frac12}},
\end{align}
where the branch cut for the square root is taken along $\gamma_1^- \cup \gamma_2^-$, i.e., the curve defined by $w_{1,-}\left((-\infty,1] \cup [1,+\infty)\right)$; see Figure \ref{fig: image of w} for an illustration. It is shown in \cite[Section 6.1.5]{Des} that the solution to RH problem \ref{rhp:global para} is given by
\begin{equation} \label{Nz-def}
N(z)= \mathcal{N}_0
\begin{pmatrix}
N_1(w_1(z)) & N_1(w_3(z)) & N_1(w_2(z))
\\
N_2(w_1(z)) & N_2(w_3(z)) & N_2(w_2(z))
\\
N_3(w_1(z)) & N_3(w_3(z)) & N_3(w_2(z))
\end{pmatrix} \Lam,
\end{equation}
where $\Lam$ is defined in \eqref{eq: MatLam} and
\begin{equation} \label{eq: MatN0}
	\mathcal{N}_0= \frac14 \diag\left(\frac{4}{2^{\frac16}},\sqrt{6}, 3\cdot 2^{\frac16} \right)
\begin{pmatrix}
\sqrt{2}i & 1 & -1
\\
0 & 2 & 2
\\
-\sqrt{2}i & 1 & -1
\end{pmatrix}
\end{equation}
is an invertible constant matrix.

It is worthwhile to mention that $\mathcal{N}_0 ^{-1}N(z)$ satisfies the following symmetric relation (see \cite[Equations (2.2.30) and (6.1.38)]{Des})
\begin{equation} \label{eq: Nz+-1}
  \mathcal{N}_0 ^{-1}N(z)  = \Lam  \mathcal{N}_0 ^{-1}N(-z) \Lam,
\end{equation}
or equivalently,
\begin{equation} \label{eq: Nz+-}
    N(z)  = \Upsilon N(-z) \Lam,
\end{equation}
where
\begin{equation} \label{eq: Upsilon-def}
  \Upsilon := \mathcal{N}_0 \Lam  \mathcal{N}_0 ^{-1}.
\end{equation}

Finally, from the asymptotic behaviors of the $w$-functions given in Proposition \ref{prop:w}, it is readily seen that the following proposition regarding the refined asymptotic behaviors of the global parametrix $N$ near $-1$ and $\infty$.
\begin{proposition}
With $N$ defined in \eqref{Nz-def}, we have
\begin{align}
  N(z) = \, & \frac{3^{\frac14}}{2^{\frac34} (z + 1)^{\frac14}}
\begin{pmatrix}
    -2^{\frac13} i & 0 & 2^{\frac13} \\
 - i & 0 & 1 \\
  2^{-\frac43} i & 0 & -2^{-\frac43}
 \end{pmatrix} + \frac{1}{2^{\frac43} \cdot 3^{\frac12}}
\begin{pmatrix}
    0 & -2^{\frac53} & 0 \\
  0 & 2^{\frac73} & 0 \\
 0 & -5 & 0
 \end{pmatrix}
\nonumber
\\
 & + \frac{(z + 1)^{\frac14}}{6^{\frac54}}
 \begin{pmatrix}
   -2^{\frac13} i & 0 & -2^{\frac13} \\
 5 i & 0 & 5 \\
 \frac{25}{2^{\frac43}}  i & 0 & \frac{25}{2^{\frac43}}
 \end{pmatrix}
  + \frac{(z + 1)^{\frac34}}{72 \cdot 6^{\frac34}}
  \begin{pmatrix}
  -13 \cdot 2^{\frac13} i & 0 & 13 \cdot 2^{\frac13} \\
 35 i & 0 & -35 \\
 -\frac{83}{2^{\frac43} } i & 0 & \frac{83}{2^{\frac43} }
 \end{pmatrix}
\nonumber
\\
& + \frac{ 2^{\frac23} (z + 1)}{27 \sqrt{3}}
\begin{pmatrix}
   0 & -2^{\frac53} & 0 \\
  0 & 2^{\frac13} & 0 \\
 0 & 7 & 0
 \end{pmatrix} + \Boh(z+1)^{\frac54},\qquad z \to -1,
\label{Nz-exp-z=-1}
\end{align}
and
\begin{equation}\label{eq:asyN}
N(z)=
\left( I + \frac{ \mathsf{N}_1}{z} + \mathcal \Boh (z^{-2}) \right) \diag \left(z^{-\frac13},1,z^{\frac13} \right)L_{\pm},\qquad z \to \infty,
\end{equation}
where
\begin{equation} \label{N1-coeff}
  \mathsf{N}_1 = \begin{pmatrix}
    0 & -2^{-\frac{5}{3}} & 0 \\
    - \frac{5}{16 \cdot 2^{\frac{1}{3}}} & 0 & 2^{ -\frac{5}{3}} \\
    0 & \frac{5}{16 \cdot 2^{\frac{1}{3}}} & 0
  \end{pmatrix},
\end{equation}
and the constants $L_{\pm}$ are given in \eqref{def:Lpm}.
\end{proposition}

\subsection{Local parametrix near $-1$}
Due to the fact that the convergence of the jump matrices to the identity matrices on $\Sigma_2^{(1)}$ and $\Sigma_4^{(1)}$ is not uniform near $-1$, we intend to find a function $P^{(-1)}(z)$ satisfying an RH problem as follows.
\begin{rhp}\label{rhp:localpara1}
We look for a $3 \times 3$ matrix-valued function $P^{(-1)}(z)$ satisfying
\begin{itemize}
\item[\rm (1)]
$P^{(-1)}(z)$ is defined and analytic in $\overline{U(-1,\delta)} \setminus \{\Sigma_2^{(1)} \cup \Sigma_3^{(1)} \cup \Sigma_4^{(1)} \}$.
\item[\rm (2)] $P^{(-1)}(z)$ satisfies the jump condition
\begin{align}\label{eq:P1-jump}
 P^{(-1)}_{+}(z)=P^{(-1)}_{-}(z)J_S(z), \qquad z\in U(-1,\delta)\cap \{\cup\Sigma_{i=2,3,4}^{(1)}\},
\end{align}
where $J_S(z)$ is defined in \eqref{def:JS}.
\item[\rm (3)]As $s \to +\infty$, $P^{(-1)}(z)$ matches $N(z)$ on the boundary $\partial U(-1,\delta)$ of $U(-1,\delta)$, i.e.,
\begin{equation}\label{eq:mathching1}
P^{(-1)}(z)=\left(I+\mathcal \Boh (s^{-\frac43} ) \right)N(z), \qquad z\in \partial U(-1,\delta).
\end{equation}
\end{itemize}
\end{rhp}

The RH problem \ref{rhp:localpara1} for $P^{(-1)}(z)$ can be solved explicitly with the aid of the Bessel parametrix $\Phi^{(\Bes)}_{\alpha}$ described in Appendix \ref{Append: BP}. To this aim, we introduce the local conformal mapping
\begin{equation}\label{def: cm f}
f(z)=\frac{1}{4}\left(\lambda_1(z)-\lambda_3(z) \right)^2, \qquad z\in U(-1,\delta).
\end{equation}
By \eqref{eq:lambda-1 at -1} and \eqref{eq:lambda-3 at -1}, we have that $f(z)$ is analytic in $U(-1,\delta)$ and, as $z\to -1$,
\begin{equation}\label{def: expf-1}
f(z)=C_1^2(z+1)+2C_1 C_3(z+1)^2 + \Boh(z+1)^3
\end{equation}
with the constants $C_1$ and $C_3$ given in \eqref{def:Ci}. Let $\Phi^{(\Bes)}_{0}$ be the Bessel parametrix in \eqref{Phi-B-solution} with $\alpha=0$, we set, for $z\in \overline{U(-1,\delta)} \setminus \{\Sigma_2^{(1)} \cup \Sigma_3^{(1)} \cup \Sigma_4^{(1)} \}$,
\begin{align}\label{eq: local par}
P^{(-1)}(z)=\, & E(z)\begin{pmatrix}
\left( \Phi^{(\Bes)}_{0} \right)_{11}(s^{\frac{8}{3}}f(z))& 0 & \left( \Phi^{(\Bes)}_{0}\right)_{12}(s^{\frac{8}{3}}f(z))
\\
0 & 1 & 0
\\
\left( \Phi^{(\Bes)}_{0}\right)_{21}(s^{\frac{8}{3}}f(z)) & 0 & \left( \Phi^{(\Bes)}_{0}\right)_{22}(s^{\frac{8}{3}}f(z))
\end{pmatrix}
\nonumber
\\
&\times \begin{pmatrix}
e^{\frac{s^{\frac43}}{2}(\lambda_3(z)-\lambda_1(z))} & 0 & 0\\
0 & 1 & 0
\\
0 & 0& e^{\frac{s^{\frac43}}{2}(\lambda_1(z)-\lambda_3(z))}
\end{pmatrix}
\nonumber
\\
& \times \left\{
                \begin{array}{ll}
                 \begin{pmatrix} 1&0&0 \\ 0&1&0 \\ 0&  e^{s^{\frac43}(\lambda_3(z)-\lambda_2(z))}  &1 \end{pmatrix},  &\quad  \hbox{$|\arg f(z)|<  \frac{3\pi}{4}  $,} \\
         I, & \quad  \hbox{$ \frac{3\pi}{4} <|\arg f(z)|<\pi$,}
                \end{array}
              \right.
\end{align}
where $f(z)$ is defined in \eqref{def: cm f} and
\begin{equation}\label{def:E}
E(z)=\frac{1}{ \sqrt{2} }N(z)\begin{pmatrix}
1 & 0 &-i \\
0 & \sqrt{2} & 0
\\
-i & 0& 1
\end{pmatrix}\begin{pmatrix}
 \pi^{\frac{1}{2}}  s^{\frac{2}{3}}f(z) ^{\frac{1}{4}}& 0 &0 \\
0 & 1 & 0
\\
0 & 0&  \pi^{-\frac{1}{2}}  s^{-\frac{2}{3}}f(z) ^{-\frac{1}{4}}
\end{pmatrix}
\end{equation}
with $N(z)$ given in \eqref{Nz-def}.

\begin{proposition}\label{prop:p-1}
The local parametrix $P^{(-1)}(z)$ defined in \eqref{eq: local par} solves the RH problem \ref{rhp:localpara1}.
\end{proposition}
\begin{proof}
From \eqref{Bessel-jump}, it is straightforward to check that $P^{(-1)}(z)$ satisfies the jump condition in \eqref{eq:P1-jump} provided the prefactor $E(z)$ is analytic in $U(-1,\delta)$. To see this, we note that the only possible jump for $E(z)$ is on the interval $(-1 - \delta , -1)$. For $x \in (-1 - \delta , -1)$, it is readily seen from \eqref{def: expf-1} that $f_+(x)^{\frac{1}{4}} = i\, f_-(x)^{\frac{1}{4}}$. Hence, we obtain from the jump of $N(z)$ in \eqref{eq:N-jump} that
\begin{align*}
  E_{-}^{-1} (x) E_+(x) = \, & \begin{pmatrix}
 \pi^{-\frac{1}{2}}  s^{-\frac{2}{3}} f_-(x) ^{-\frac{1}{4}}& 0 &0 \\
0 & 1 & 0
\\
0 & 0&  \pi^{\frac{1}{2}}  s^{\frac{2}{3}} f_-(x) ^{-\frac{1}{4}}
\end{pmatrix} \begin{pmatrix}
\frac{1}{\sqrt{2}} & 0  & \frac{i}{\sqrt{2}}  \\
0 & 1 & 0
\\
\frac{i}{\sqrt{2}}& 0& \frac{1}{\sqrt{2}}
\end{pmatrix} \begin{pmatrix} 0&0&1 \\ 0&1&0 \\ -1&0&0 \end{pmatrix}\\
& \times  \begin{pmatrix}
\frac{1}{\sqrt{2}} & 0  & -\frac{i}{\sqrt{2}}  \\
0 & 1 & 0
\\
-\frac{i}{\sqrt{2}}& 0& \frac{1}{\sqrt{2}}
\end{pmatrix} \begin{pmatrix}
 \pi^{\frac{1}{2}}  s^{\frac{2}{3}}f_+(x) ^{\frac{1}{4}}& 0 &0 \\
0 & 1 & 0
\\
0 & 0&  \pi^{-\frac{1}{2}}  s^{-\frac{2}{3}} f_+(x) ^{-\frac{1}{4}}
\end{pmatrix}=I.
\end{align*}
This shows that $E(z)$ is indeed analytic in $U(-1,\delta)$.

It remains to verify the matching condition \eqref{eq:mathching1}. From the local behaviors of $\lambda_1$ and $\lambda_3$ near $-1$ given in item (d) of Proposition \ref{prop: prop of lambda}, the function $e^{\frac{s^{\frac43}}{2}(\lambda_1(z)-\lambda_3(z))}$ appearing in the last term of \eqref{eq: local par} is exponentially small as $s\to +\infty$ for $z \in \partial U(-1,\delta)$. Thus, it follows from the asymptotic behavior of the Bessel parametrix $\Phi^{(\Bes)}_{0}$ at infinity in \eqref{eq:Besl-infty} that, for $z \in \partial U(-1,\delta)$,
\begin{eqnarray} \label{eq:mathching1-exp}
  P^{(-1)}(z) N^{-1}(z)=  I+ \frac{J_{1}^{(-1)} (z) }{ s^{\frac{4}{3}} } + \mathcal \Boh (s^{-\frac83} ) ,  \qquad  s \to +\infty,
\end{eqnarray}
where
\begin{equation} \label{j1-1-formula}
 J_{1}^{(-1)} (z) = \frac{1}{8  f(z)^{\frac{1}{2}}} N(z) \begin{pmatrix}
  -1 & 0 & -2i \\ 0 & 0 & 0 \\ -2i & 0 & 1
  \end{pmatrix} N^{-1}(z),
\end{equation}
which gives us \eqref{eq:mathching1}.

This completes the proof of Proposition \ref{prop:p-1}.
\end{proof}

We conclude this section by evaluating $E(-1)$ and $E'(-1)$ for later use. The calculations are straightforward and cumbersome by combining
\eqref{def:E} and the asymptotics of $N(z)$ and $f(z)$ given in \eqref{Nz-exp-z=-1} and \eqref{def: expf-1}. We omit the details but present the results below.
\begin{equation} \label{eq: E(-1)}
  E(-1) = \begin{pmatrix}
    -i  2^{\frac{1}{12}} 3^{\frac14} C_1^{\frac12} \pi^{\frac12}  s^{\frac23} & - \frac{2^{\frac13}}{3^{\frac12}} &
    - \frac{1}{2^{\frac5{12}} 3^{\frac54} C_1^{\frac12} \pi^{\frac12} s^{\frac23} } \\
   - i 2^{-\frac14} 3^{\frac14} C_1^{\frac12} \pi^{\frac12} s^{\frac23} & \frac{2}{3^{\frac12}}  &
   \frac{5}{ 2^{\frac34} 3^{\frac54} C_1^{\frac12} \pi^{\frac12} s^{\frac23}} \\
   i\frac{3^{\frac14} C_1^{\frac12} \pi^{\frac12} s^{\frac23} }{
  2^{\frac{19}{12}}} & -\frac{5}{2^{\frac43} 3^{\frac12}} & \frac{25}{
  2^{\frac{25}{12}} 3^{\frac54}  C_1^{\frac12} \pi^{\frac12} s^{\frac23}}
  \end{pmatrix},
\end{equation}
and
\begin{equation} \label{eq: Ep(-1)}
  E'(-1) = \begin{pmatrix}
    -i\frac{  (13 C_1 + 108 C_3) \pi^{\frac12} s^{\frac23}}{36 \cdot 2^{\frac{11}{12}} 3^{\frac34}  C_1^{\frac12}} & * &  * \\
    i \frac{(35 C_1 - 108 C_3)  \pi^{\frac12} s^{\frac23}}{72 \cdot 2^{\frac14} 3^{\frac34}  C_1^{\frac12}} & * & * \\
   -i \frac{(83 C_1 - 108 C_3) \pi^{\frac12} s^{\frac23}}{144 \cdot 2^{\frac{7}{12}} 3^{\frac34}  C_1^{\frac12}} & * & *
  \end{pmatrix},
\end{equation}
where the constants $C_1$ and $C_3$ are given in \eqref{def:Ci}, and $*$ stands for some unimportant entry.

\subsection{Local parametrix near $1$}
Similar to the situation encountered near $z=-1$, we intend to find a function $P^{(1)}(z)$ satisfying the following RH problem near $z=1$.

\begin{rhp}\label{rhp:localpara2}
We look for a $3 \times 3$ matrix-valued function $P^{(1)}(z)$ satisfying
\begin{itemize}
\item[\rm (1)]
$P^{(1)}(z)$ is defined and analytic in $\overline{U(1,\delta)} \setminus \{\Sigma_0^{(1)} \cup \Sigma_1^{(1)} \cup \Sigma_5^{(1)}\}$.
\item[\rm (2)] $P^{(1)}(z)$ satisfies the jump condition
\begin{equation}\label{eq:P2-jump}
 P^{(1)}_{+}(z)=P^{(1)}_{-}(z)J_S(z), \qquad z\in U(-1,\delta)\cap \{\cup\Sigma_{i=0,1,5}^{(1)}\},
\end{equation}
where $J_S(z)$ is defined in \eqref{def:JS}.
\item[\rm (3)]As $s \to +\infty$, we have
\begin{equation}\label{eq:mathching2}
P^{(1)}(z)=\left(I+\mathcal \Boh (s^{-\frac43} ) \right)N(z),\qquad z\in \partial U(1,\delta).
\end{equation}
\end{itemize}
\end{rhp}
Again, the RH problem \ref{rhp:localpara2} can be solved with the help of the Bessel parametrix $\Phi^{(\Bes)}_{0}$, following the same spirit
in the construction of $P^{(-1)}(z)$.
The conformal mapping now reads
\begin{equation}\label{def: cm f-2}
\widetilde{f}(z)=\frac{1}{4} ( \lambda_1(z)-\lambda_2(z) )^2,\qquad z\in  U(1,\delta).
\end{equation}
From \eqref{eq: lambda1z+-} and \eqref{eq: lambda23z}, one can see that
\begin{equation} \label{eq: fz+-}
  \widetilde{f}(z) = f(-z),
\end{equation}
where $f(z)$ is defined in \eqref{def: cm f}. In view of \eqref{def: expf-1}, this also implies that, as $z\to 1$,
\begin{equation}\label{def: expf-2}
\widetilde f(z)= - C_1^2(z-1)+2C_1 C_3(z-1)^2+\Boh(z-1)^3,
\end{equation}
with the constants $C_1$ and $C_3$ given in \eqref{def:Ci}. For $z\in \overline{U(1,\delta)} \setminus \{\cup\Sigma_{i=0,1,5}^{(1)}\}$,
we then set
\begin{align}\label{eq: local par-2}
P^{(1)}(z)=&\, \widetilde E(z)\begin{pmatrix}
\left( \Phi^{(\Bes)}_{0} \right)_{11}(s^{\frac{8}{3}} \widetilde f(z))& -\left( \Phi^{(\Bes)}_{0} \right)_{12}(s^{\frac{8}{3}}\widetilde f(z)) & 0
\\
-\left( \Phi^{(\Bes)}_{0} \right)_{21}(s^{\frac{8}{3}}\widetilde f(z)) & \left( \Phi^{(\Bes)}_{0} \right)_{22}(s^{\frac{8}{3}}\widetilde f(z)) & 0
\\
0 & 0 & 1
\end{pmatrix}
\nonumber
\\
& \times
\begin{pmatrix}
e^{\frac{s^{\frac43}}{2}(\lambda_2(z)-\lambda_1(z))} & 0 & 0
\\
0 & e^{\frac{s^{\frac43}}{2}(\lambda_1(z)-\lambda_2(z))} & 0 \\
0 & 0 & 1
\end{pmatrix}
\nonumber
\\
& \times \left\{
                \begin{array}{ll}
                 \begin{pmatrix} 1&0&0 \\ 0&1& e^{s^{\frac43}(\lambda_2(z)-\lambda_3(z))} \\ 0&  0  &1 \end{pmatrix},  &\quad  \hbox{$|\arg \widetilde f(z)|<  \frac{3\pi}{4}  $,} \\
         I, & \quad  \hbox{$  \frac{3\pi}{4} <|\arg \widetilde f(z)|<\pi$,}
                \end{array}
              \right.
\end{align}
where $\widetilde f(z)$ is defined in \eqref{def: cm f-2} and
\begin{equation}\label{def:E-2}
\widetilde E(z)=\frac{1}{ \sqrt{2} }N(z)\begin{pmatrix}
1 & i & 0
\\
i & 1 & 0 \\
0 & 0 & \sqrt{2}
\end{pmatrix}\begin{pmatrix}
 \pi^{\frac{1}{2}}  s^{\frac{2}{3}} \widetilde f(z) ^{\frac{1}{4}}& 0 & 0 \\
0 & \pi^{-\frac{1}{2}}  s^{-\frac{2}{3}} \widetilde f(z) ^{-\frac{1}{4}} & 0 \\
0 & 0 & 1
\end{pmatrix}.
\end{equation}

It comes out that $P^{(1)}(z)$ is closely related to $P^{(-1)}(z)$. To see the relation, we observe from \eqref{eq: Nz+-} and \eqref{eq: fz+-} that
\begin{equation}
  \widetilde E(z) = \Upsilon   E(-z)  \Lam,
\end{equation}
where the constant matrices $\Lam$ and $\Upsilon$ are given in \eqref{eq: MatLam} and \eqref{eq: Upsilon-def}, respectively.
A further appeal to the symmetric relations \eqref{eq: lambda1z+-} and \eqref{eq: lambda23z} then implies
\begin{equation}
  P^{(1)}(z) = \Upsilon    P^{(-1)}(-z)  \Lam.
\end{equation}
This in turn shows that the $ P^{(1)}(z)$ fulfills the jump condition \eqref{eq:P2-jump} and the matching condition \eqref{eq:mathching2}. In particular, we have, as  $s \to +\infty$,
\begin{eqnarray} \label{eq:mathching2-exp}
  P^{(1)}(z) N^{-1}(z)=  I+ \frac{J_{1}^{(1)} (z) }{ s^{\frac{4}{3}} } + \mathcal \Boh (s^{-\frac83} ) ,  \qquad z \in \partial U(1,\delta),
\end{eqnarray}
where
\begin{equation} \label{j1-2-formula}
 J_{1}^{(1)} (z) = \frac{1}{8  \widetilde f(z)^{\frac{1}{2}}}  N(z)
\begin{pmatrix}
-1 & 2i & 0 \\ 2i & 1 & 0 \\ 0 & 0 & 0
\end{pmatrix} N^{-1}(z).
\end{equation}
Again, using \eqref{eq: Nz+-} and \eqref{eq: fz+-}, it is readily seen from \eqref{j1-1-formula} that
\begin{align} \label{j1-symmetric}
   J_{1}^{(1)} (z) 
&=\frac{1}{8   f(-z)^{\frac{1}{2}}} \Upsilon N(-z) \Lambda \begin{pmatrix}
-1 & 2i & 0 \\ 2i & 1 & 0 \\ 0 & 0 & 0
\end{pmatrix} \Lambda^{-1} N^{-1}(-z) \Upsilon^{-1}
\nonumber
\\
&=\Upsilon J_1^{(-1)}(-z) \Upsilon,
\end{align}
where we have made use of the fact that $\Upsilon^{-1}=\Upsilon$; see \eqref{eq: Upsilon-def}. In summary, we have proved the following proposition.
\begin{proposition}\label{prop:p-2}
The local parametrix $P^{(1)}(z)$ defined in \eqref{eq: local par-2} solves the RH problem \ref{rhp:localpara2}.
\end{proposition}

\subsection{Final transformation}
Our final transformation is defined by
\begin{equation}\label{def:StoR}
 R(z)=\left\{
                \begin{array}{ll}
 S(z)N^{-1}(z), &\quad  \hbox{$z\in \mathbb{C}\setminus \{U(-1,\delta) \cup U(1,\delta)\cup  \Sigma_S\}$,} \\
        S(z)(P^{(-1)}(z))^{-1},        &\quad  \hbox{$z \in U(-1,\delta)$,} \\
        S(z)(P^{(1)}(z))^{-1},  &\quad  \hbox{$z \in U(1,\delta)$.}
                \end{array}
              \right.
\end{equation}
It is then easily seen that $R(z)$ satisfies the following RH problem.

\begin{rhp}\label{rhp:R}
The $3 \times 3$ matrix-valued function $R(z)$ defined in \eqref{def:StoR} has the following properties:
\begin{itemize}
\item[\rm (1)]
$R(z)$ is analytic in $\mathbb{C}\setminus \Sigma_R$, where the contour $\Sigma_R$ is shown in Figure \ref{fig:ContourR}.
\item[\rm (2)] $R(z)$ satisfies the jump condition
\begin{equation}\label{eq:R-jump}
 R_{+}(z)=R_{-}(z)J_R(z),\qquad z\in\Sigma_R,
\end{equation}
where
\begin{equation}
J_R(z)=
 \left\{
                \begin{array}{ll}
 P^{(-1)}(z)N^{-1}(z),        &\quad  \hbox{$z\in \partial U(-1,\delta)$,} \\
         P^{(1)}(z)N^{-1}(z), &\quad  \hbox{$z\in \partial U(1,\delta)$,}\\
       N(z)J_S(z)N^{-1}(z),   &\quad  \hbox{$z \in \cup\Sigma_{i=1,2,4,5}^{(1)}\setminus \{U(-1,\delta)\cup U(1,\delta)\}$.}
                \end{array}
              \right.
\end{equation}
\item[\rm (3)]As $z \to \infty$, we have
\begin{equation}\label{eq:asyR1}
R(z)=I+\Boh(z^{-1}).
\end{equation}
\end{itemize}
\end{rhp}

\begin{figure}[t]
\begin{center}
   \setlength{\unitlength}{1truemm}
   \begin{picture}(100,70)(-5,2)
       \put(37.5,41){\line(-2,1){28}}
       \put(37.5,39){\line(-2,-1){28}}
       \put(62.5,41){\line(2,1){28}}
       \put(62.5,39){\line(2,-1){28}}
       \put(40,40){\thicklines\circle*{1}}
       \put(40,40){\circle{5}}

       \put(60,40){\thicklines\circle*{1}}
       \put(60,40){\circle{5}}

       \put(41,42.3){\thicklines\vector(1,0){.0001}}
       \put(61,42.3){\thicklines\vector(1,0){.0001}}
       \put(38,34){$-1$}
       \put(59,34){$1$}


       \put(20,50){\thicklines\vector(2,-1){.0001}}
       \put(20,30.5){\thicklines\vector(2,1){.0001}}
       \put(80,30.5){\thicklines\vector(2,-1){.0001}}
       \put(80,50){\thicklines\vector(2,1){.0001}}

  \end{picture}
   \vspace{-16mm}
   \caption{Contour  $\Sigma_R$ for the RH problem \ref{rhp:R} for $R$.}
   \label{fig:ContourR}
\end{center}
\end{figure}
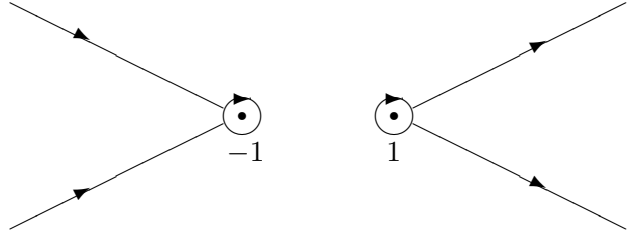

In view of Corollary \ref{coro:JS}, it follows that the jumps of $R(z)$ tend to the identity matrix exponentially fast as $s\to +\infty$, except for those on $\partial U(-1,\delta)\cup \partial U(1,\delta)$; see also \eqref{eq:mathching1-exp} and \eqref{eq:mathching2-exp} for the expansions of $J_R(z)$ on $\partial U(1,\delta)$ and $\partial U(-1,\delta)$.
Then, by a standard argument (cf. \cite{DeiftBook}), we conclude that, as $s\to +\infty$,
\begin{equation}\label{eq:estR}
R(z)=I+ \frac{R_1(z)}{s^{\frac{4}{3}}}+ \Boh(s^{-\frac{8}{3}}) \quad \textrm{and} \quad \frac{\ud}{\ud z}R(z)=\frac{R_1'(z)}{s^{\frac{4}{3}}}+ \Boh(s^{-\frac{8}{3}}),
\end{equation}
uniformly for $z \in \mathbb{C} \setminus \Sigma_R$. Moreover, a combination of \eqref{eq:estR} and RH problem \ref{rhp:R} shows that $R_1(z)$ solves the following RH problem.
\begin{rhp}\label{rhp:R1 z}
The $3 \times 3$ matrix-valued function $R_1(z)$ appearing in \eqref{eq:estR} satisfies the following properties:
\begin{itemize}
\item[\rm (1)]
$R_1(z)$ is defined and analytic in $\mathbb{C}\setminus \{\partial U(-1,\delta)\cup \partial U(1,\delta)\}$.
\item[\rm (2)] $R_1$ satisfies the jump condition
\begin{equation}\label{eq:R1-jump}
 R_{1,+}(z)-R_{1,-}(z) = \begin{cases} J_{1}^{(-1)}(z), & z \in \partial U(-1,\delta),  \\
 J_{1}^{(1)}(z), & z \in \partial U(1,\delta),  \end{cases}
\end{equation}
where $J_{1}^{(-1)}(z)$ and $J_{1}^{(1)}(z)$ are given in \eqref{j1-1-formula} and \eqref{j1-2-formula}, respectively.

\item[\rm (3)]As $z \to \infty$, we have $R_1(z) = \Boh(z^{-1}).$
\end{itemize}
\end{rhp}

By Cauchy's residue theorem, the solution to the above RH problem is given by
\begin{align}
& R_1(z)
\nonumber
\\
&= \frac{1}{2 \pi i} \oint_{\partial U(-1,\delta)} \frac{J_{1}^{(-1)}(\zeta)}{z-\zeta } \ud \zeta + \frac{1}{2 \pi i} \oint_{\partial U(1,\delta)} \frac{J_{1}^{(1)}(\zeta)}{z- \zeta} \ud \zeta \nonumber  \\
&= \begin{cases}
\ds \frac{\mathop{\textrm{Res}}\limits_{\zeta = -1}J_{1}^{(-1)}(\zeta)}{z+1} + \frac{\mathop{\Res}\limits_{\zeta = 1}J_{1}^{(1)}(\zeta)}{z-1} , & z \in \mathbb{C}\setminus \{U(-1,\delta)\cup  U(1,\delta)\}, \medskip \\
\ds \frac{\mathop{\textrm{Res}}\limits_{\zeta = -1} J_{1}^{(-1)}(\zeta)}{z+1}  + \frac{\mathop{\Res}\limits_{\zeta = 1}J_{1}^{(1)}(\zeta)}{z-1} - J_{1}^{(-1)}(z) , & z \in U(-1,\delta), \medskip \\
\ds \frac{\mathop{\Res}\limits_{\zeta = -1}J_{1}^{(-1)}(\zeta)}{z+1}  + \frac{\mathop{\Res}\limits_{\zeta = 1}J_{1}^{(1)}(\zeta)}{z-1} - J_{1}^{(1)}(z), & z \in U(1,\delta).
\end{cases}
\label{R1-expression}
\end{align}

For our purpose, we need to know the exact value of $R_1'(-1)$. From the local behaviors of $N(z)$ and $f(z)$ near $z=-1$ given in \eqref{Nz-exp-z=-1} and \eqref{def: expf-1}, we obtain from \eqref{j1-1-formula} that
\begin{align}\label{J1-1-expan}
  J_{1}^{(-1)}(z)
  &= \frac{1}{8  f(z)^{\frac{1}{2}}} N(z)
  \begin{pmatrix}
  -1 & 0 & -2i \\ 0 & 0 & 0 \\ -2i & 0 & 1
  \end{pmatrix} N^{-1}(z)
  \nonumber \\
  &=\frac{\mathcal{J}_{-1}}{z+1} + \mathcal{J}_0 + \mathcal{J}_1 \cdot (z+1) + \Boh(z+1)^2, \qquad  z \to -1,
\end{align}
where
\begin{equation} \label{j1-1-residue}
  \mathcal{J}_{-1}= \Res_{\zeta = -1}J_{1}^{(-1)}(\zeta) = \frac{1}{16 \sqrt{6} \, C_1} \begin{pmatrix}
    1 & - 2^{\frac43} & - 2^{\frac53} \\
    2^{-\frac13} & -2 & -2^{\frac43} \\
    -2^{-\frac53} & 2^{-\frac13} & 1
  \end{pmatrix}
\end{equation}
with $C_1$ given in \eqref{def:Ci}. Although the explicit formulas of $\mathcal{J}_0$ and $\mathcal{J}_1$ in \eqref{J1-1-expan} are also available and $\mathcal{J}_{1}$ is indeed involved in our later calculation, we decide not to include the exact formulas here due to their complicated forms. Also note that $J_1^{(1)}(z)=\Upsilon J_1^{(-1)}(-z) \Upsilon$ (see \eqref{j1-symmetric}), it is then readily seen that
\begin{align}
  \Res_{\zeta = 1}J_{1}^{(1)}(\zeta) &= - \Upsilon \, \Res_{\zeta = -1}J_{1}^{(-1)}(\zeta) \, \Upsilon = -\Upsilon \, \mathcal{J}_{-1} \, \Upsilon \nonumber \\
  & =  \frac{1}{16 \sqrt{6} C_1} \begin{pmatrix}
    -1 &  -2^{\frac43} & 2^{\frac53}  \\
    2^{-\frac13} & 2 & - 2^{\frac43} \\
    2^{-\frac53} & 2^{-\frac13} & -1
  \end{pmatrix}. \label{j1-2-residue}
\end{align}
Hence, in view of \eqref{R1-expression}, we arrive at
\begin{equation} \label{Rp(-1)-expression}
  R_1'(-1) = - \mathcal{J}_{1} - \frac{1}{4} \Res_{\zeta = 1}J_{1}^{(1)}(\zeta)= -\mathcal{J}_{1}-\frac{1}{64 \sqrt{6} C_1}
  \begin{pmatrix}
    -1 &  -2^{\frac43} & 2^{\frac53}  \\
    2^{-\frac13} & 2 & - 2^{\frac43} \\
    2^{-\frac53} & 2^{-\frac13} & -1
  \end{pmatrix},
\end{equation}
where $\mathcal{J}_{1}$ is given in \eqref{J1-1-expan}.

Finally, we point out that,
\begin{equation}\label{eq:asyR}
R(z)=I + \frac{\mathsf{R}_1}{z} +\Boh(z^{-2}), \qquad z\to \infty,
\end{equation}
where
$$
\mathsf{R}_1=\frac{i}{2 \pi}\int_{\Sigma_R} R_-(w)\left(J_{ R}(w)-I\right)\ud w.
$$
Comparing  \eqref{eq:estR}, \eqref{R1-expression} with \eqref{eq:asyR}, we obtain
\begin{align} \label{R1-coeff}
\mathsf{R}_1 &=\frac{1}{s^{\frac{4}{3}}} \left(\Res_{\zeta = -1}J_{1}^{(-1)}(\zeta)+\Res_{\zeta = 1}J_{1}^{(1)}(\zeta)\right) + \Boh(s^{-\frac{8}{3}})
\nonumber
\\
&= \frac{1}{s^{\frac{4}{3}}} \left(\mathcal{J}_{-1} - \Upsilon \mathcal{J}_{-1}  \Upsilon \right) + \Boh(s^{-\frac{8}{3}}),
\end{align}
on account of \eqref{j1-1-residue} and \eqref{j1-2-residue}.

We are now ready to prove Theorem \ref{main-thm}.

\section{Proof of Theorem \ref{main-thm}}\label{sec:proof}
Our strategy is to find the large $s$ asymptotics of $\frac{\partial}{\partial s} F(s;\rho)$ and $\frac{\partial}{\partial \rho} F(s;\rho)$ by making use of the differential identities established in Proposition \ref{prop:derivativeandX}, which will in turn give us the asymptotics of $F(s;\rho)$.

We start with deriving the asymptotics of $\frac{\partial}{\partial s} F(s;\rho)$.
Substituting \eqref{def:XtoT}  into the differential identity  \eqref{eq:derivativeinsX-2} gives us
\begin{align}
\frac{\partial}{\partial s} F(s;\rho)&=-\frac{1}{ \pi i s} \lim_{z \to -1} \left[\left(T^{-1}(z)T'(z)\right)_{21}+\left(T^{-1}(z) T'(z)\right)_{31}\right],
\label{eq:derivativeinsT}
\end{align}
where $'$ denotes the derivative with respect to $z$. Tracing back the invertible transformations $T\mapsto S$ and $S\mapsto R$ in \eqref{def:TtoS} and \eqref{def:StoR}, it follows that
\begin{multline}\label{eq:TtoR}
T(z)=e^{D_0s^{\frac43}} \sqrt{\frac{2\pi}{3}}i e^{\frac{\rho^2}{6}} \Psi_0 \diag \left(s^{-\frac13},1,s^{\frac13} \right) S_0^{-1}
\\ \times R(z)P^{(-1)}(z)\diag\left(e^{s^{\frac43}\lambda_1(z)},e^{s^{\frac43}\lambda_2(z)},e^{s^{\frac43}\lambda_3(z)}\right), \qquad z \in U(-1,\delta).
\end{multline}
In addition, from the explicit expression of $P^{(-1)}(z)$ in \eqref{eq: local par}, we further obtain
\begin{align} \label{Tformula-A-B}
T(z)&=e^{D_0s^{\frac43}} \sqrt{\frac{2\pi}{3}}i e^{\frac{\rho^2}{6}} \Psi_0 \diag \left(s^{-\frac13},1,s^{\frac13} \right) S_0^{-1}  R(z)E(z) \mathscr{B}(s^{\frac{8}{3}}f(z)) \mathscr{A}(z)
\end{align}
for $z \in U(-1,\delta)$ and $|\arg f(z)|<\frac34\pi$, where
\begin{equation} \label{eq: A(z) def}
	\mathscr{A}(z) := \begin{pmatrix}
e^{\frac{s^{\frac43}}{2}(\lambda_1(z)+\lambda_3(z))} & 0 & 0\\
0 & e^{s^{\frac43}\lambda_2(z)} & 0
\\
0 &  e^{\frac{s^{\frac43}}{2}(\lambda_1(z)+\lambda_3(z))} & e^{\frac{s^{\frac43}}{2}(\lambda_1(z)+\lambda_3(z))}
\end{pmatrix}
\end{equation}
and
\begin{equation}\label{def:Bz}
	\mathscr{B}(z):=\begin{pmatrix}
\left( \Phi^{(\Bes)}_{0}\right)_{11}(z)& 0 & \left( \Phi^{(\Bes)}_{0}\right)_{12}(z)
\\
0 & 1 & 0
\\
\left( \Phi^{(\Bes)}_{0}\right)_{21}(z) & 0 & \left( \Phi^{(\Bes)}_{0} \right)_{22}(z)
\end{pmatrix}.
\end{equation}
Since the prefactor $e^{D_0s^{\frac43}} \sqrt{\frac{2\pi}{3}}i e^{\frac{\rho^2}{6}} \Psi_0 \diag \left(s^{-\frac13},1,s^{\frac13} \right) S_0^{-1}$ in \eqref{Tformula-A-B} is independent of $z$, we have
\begin{align}
	T^{-1}(z)T'(z)  = \, & \mathscr{A}^{-1}(z) \mathscr{A}'(z) + s^{\frac{8}{3}}f'(z) \mathscr{A}^{-1}(z) \mathscr{B}^{-1}(s^{\frac{8}{3}}f(z)) \mathscr{B}'(s^{\frac{8}{3}}f(z)) \mathscr{A}(z) \nonumber \\
 &+ \mathscr{A}^{-1}(z) \mathscr{B}^{-1}(s^{\frac{8}{3}}f(z)) E^{-1}(z) E'(z) \mathscr{B}(s^{\frac{8}{3}}f(z)) \mathscr{A}(z)\nonumber \\
	& + \mathscr{A}^{-1}(z) \mathscr{B}^{-1}(s^{\frac{8}{3}}f(z)) E^{-1}(z) R^{-1}(z) R'(z) E(z) \mathscr{B}(s^{\frac{8}{3}}f(z)) \mathscr{A}(z). \label{T'formula-A-B}
\end{align}

We next evaluate the four terms on the right hand side of the above formula one by one. For the first term in \eqref{T'formula-A-B}, we see from \eqref{eq: A(z) def} that
\begin{eqnarray} \label{Term1-final}
  \mathscr{A}^{-1}(z) \mathscr{A}'(z) = \displaystyle \frac{s^{\frac{4}{3}}}{2} \begin{pmatrix}
    \lambda_1'(z) + \lambda_3'(z) & 0 & 0 \\
    0 & 2\lambda_2'(z) & 0 \\
   0 & \lambda_1'(z) +\lambda_3'(z) - 2\lambda_2'(z)  & \lambda_1'(z)+ \lambda_3'(z)  \end{pmatrix}.
\end{eqnarray}
Thus,
\begin{equation}\label{eq:term1}
\left(\mathscr{A}^{-1}(z) \mathscr{A}'(z)\right)_{21}=\left(\mathscr{A}^{-1}(z) \mathscr{A}'(z)\right)_{31}=0.
\end{equation}
For the second term in \eqref{T'formula-A-B}, we recall the following properties of the modified  Bessel functions (cf. \cite[Chapter 10]{DLMF}):
\begin{align*}
  I_0(z)&= \sum_{k=0}^{\infty}\frac{(z/2)^{2k}}{(k!)^2}, \\
  K_0(z)&= -\left(\ln(z/2)+\gamma\right) I_0(z)+\Boh(z^2), \qquad z \to 0,
\end{align*}
where $\gamma$ is the Euler's constant. This, together with \eqref{def:Bz} and \eqref{Phi-B-solution}, implies that, as $z \to 0$,
\begin{equation}\label{eq:Bzero}
\mathscr{B}(z)=
\begin{pmatrix}
1+\Boh(z)& 0 &  \Boh(\ln z) \\
0 & 1 & 0
\\
\frac{\pi i}{2}z+\Boh(z^2) & 0 &1+\Boh(z \ln z)
\end{pmatrix}
\end{equation}
and
\begin{multline}\label{eq:B-1zero}
\mathscr{B}^{-1}(z)=\begin{pmatrix}
\left( \Phi^{(\Bes)}_{0}\right)_{22}(z)& 0 & -\left( \Phi^{(\Bes)}_{0} \right)_{12}(z)
\\
0 & 1 & 0
\\
-\left( \Phi^{(\Bes)}_{0}\right)_{21}(z) & 0 & \left( \Phi^{(\Bes)}_{0} \right)_{11}(z)
\end{pmatrix}
\\
=
\begin{pmatrix}
1+\Boh(z\ln z)& 0 &  \Boh(\ln z) \\
0 & 1 & 0
\\
-\frac{\pi i}{2}z +\Boh(z^2) & 0 &1+\Boh(z)
\end{pmatrix}.
\end{multline}
A combination of these two formulas shows
\begin{equation}\label{eq:B-1B}
\lim_{z\to 0}\left( \mathscr{B}^{-1}(z)\mathscr{B}'(z) \right)_{21}=0, \qquad \lim_{z\to 0}\left( \mathscr{B}^{-1}(z)\mathscr{B}'(z) \right)_{31} =\frac{\pi i}{2}.
\end{equation}
To this end, we note that for an arbitrary $3\times 3$ matrix $ M = (m_{ij})_{i,j=1}^3$, it is readily seen from \eqref{eq: A(z) def} that
\begin{equation} \label{eq: A-1MA}
  \begin{split}
    &\lim_{z \to -1} \left(\mathscr{A}^{-1}(z)  M  \mathscr{A}(z) \right)_{21} =  m_{21} e^{\frac{s^{\frac{4}{3}}}{2} [\lambda_1(-1) + \lambda_3(-1) -2 \lambda_2(-1)]}, \\
  &\lim_{z \to -1} \left(\mathscr{A}^{-1}(z) M  \mathscr{A}(z) \right)_{31} = m_{31} - m_{21} e^{\frac{s^{\frac{4}{3}}}{2} [\lambda_1(-1) + \lambda_3(-1) -2 \lambda_2(-1)]}.
  \end{split}
\end{equation}
We then obtain from \eqref{eq:B-1B}, \eqref{eq: A-1MA} and the facts $f(-1) = 0, f'(-1) = C_1^2$ (see \eqref{def: expf-1}) that
\begin{align} \label{Term2-final}
&\lim_{z \to -1}  \left(s^{\frac{8}{3}}f'(z) \mathscr{A}^{-1}(z) \mathscr{B}^{-1}(s^{\frac{8}{3}}f(z)) \mathscr{B}'(s^{\frac{8}{3}}f(z)) \mathscr{A}(z) \right)_{21} = 0,
\\
\label{Term2-final2}
&\lim_{z \to -1}  \left(s^{\frac{8}{3}}f'(z) \mathscr{A}^{-1}(z) \mathscr{B}^{-1}(s^{\frac{8}{3}}f(z)) \mathscr{B}'(s^{\frac{8}{3}}f(z)) \mathscr{A}(z) \right)_{31} = \frac{C_1^2 \pi i}{2} s^{\frac{8}{3}},
\end{align}
where $C_1$ is given in \eqref{def:Ci}.


Regrading the last two terms in \eqref{T'formula-A-B}, we first observe from items (d) and (f) of Proposition \ref{prop: prop of lambda}
that
\begin{align}
\lambda_1(-1) + \lambda_3(-1) - 2 \lambda_2(-1)
&=
\lambda_1(-1) + \lambda_3(-1)
-2\left( - \frac{9}{2^{\frac{7}{3}}} + \frac{3 \rho}{2^{\frac{2}{3}} s^{\frac{2}{3}}} - \lambda_1(-1) - \lambda_3(-1)\right) \nonumber
\\
&=6C_0+\frac{9}{2^{\frac43}}-\frac{3\cdot 2^{\frac13} \rho}{s^{\frac23}}= -\frac{9}{2^{\frac73}} - \frac{3 \rho}{2^{\frac23} s^{\frac23}}. \label{lambda-123-value}
\end{align}
This means the exponential terms on the right hand side of \eqref{eq: A-1MA} are exponentially small as $s \to + \infty$ for $\rho$ in any compact subset of $\mathbb{R}$. Next, it is readily seen from \eqref{eq:Bzero}, \eqref{eq:B-1zero} and \eqref{def: expf-1} that, for an arbitrary $3\times 3$ matrix $ M = (m_{ij})_{i,j=1}^3$,
\begin{equation} \label{eq: B-1MB}
\begin{split}
  &\lim_{z \to -1} \left(\mathscr{B}^{-1}(s^{\frac{8}{3}}f(z)) M  \mathscr{B}(s^{\frac{8}{3}}f(z)) \right)_{21} = m_{21},
   \\
  &\lim_{z \to -1} \left(\mathscr{B}^{-1}(s^{\frac{8}{3}}f(z)) M  \mathscr{B}(s^{\frac{8}{3}}f(z)) \right)_{31} = m_{31}.
\end{split}
\end{equation}
This, together with \eqref{eq: A-1MA} and \eqref{lambda-123-value}, implies that
\begin{align}
&\lim_{z\to -1}\left(\mathscr{A}^{-1}(z) \mathscr{B}^{-1}(s^{\frac{8}{3}}f(z)) E^{-1}(z) E'(z) \mathscr{B}(s^{\frac{8}{3}}f(z)) \mathscr{A}(z)\right)_{21}  = \Boh(e^{-( 9 \cdot 2^{-\frac{10}{3}} - \delta ) s^{\frac43}}), \nonumber \\
& \lim_{z\to -1}\left( \mathscr{A}^{-1}(z) \mathscr{B}^{-1}(s^{\frac{8}{3}}f(z)) E^{-1}(z) R^{-1}(z) R'(z) E(z) \mathscr{B}(s^{\frac{8}{3}}f(z)) \mathscr{A}(z)\right)_{21}    \\
& \hspace{10cm} =\Boh(e^{-( 9 \cdot 2^{-\frac{10}{3}} - \delta ) s^{\frac43}}) \nonumber
\end{align}
for an arbitrary small constant $\delta>0$, i.e., the above two terms are exponentially small as $s \to +\infty$.
Similarly, we also have
\begin{multline}\label{eq:3term1}
\lim_{z\to -1}\left(\mathscr{A}^{-1}(z) \mathscr{B}^{-1}(s^{\frac{8}{3}}f(z)) E^{-1}(z) E'(z) \mathscr{B}(s^{\frac{8}{3}}f(z)) \mathscr{A}(z)\right)_{31}
\\
=\lim_{z \to -1}\left( E^{-1}(z) E'(z) \right)_{31} + \Boh(e^{-( 9 \cdot 2^{-\frac{10}{3}} - \delta ) s^{\frac43}}),
\end{multline}
\begin{multline}\label{eq:4term1}
\lim_{z\to -1}\left( \mathscr{A}^{-1}(z) \mathscr{B}^{-1}(s^{\frac{8}{3}}f(z)) E^{-1}(z) R^{-1}(z) R'(z) E(z) \mathscr{B}(s^{\frac{8}{3}}f(z)) \mathscr{A}(z) \right)_{31}
\\
=\lim_{z \to -1}\left( E^{-1}(z) R^{-1}(z) R'(z) E(z) \right)_{31}+ \Boh(e^{-( 9 \cdot 2^{-\frac{10}{3}} - \delta ) s^{\frac43}}).
\end{multline}
For the limit of the right hand side on of \eqref{eq:3term1}, it follows from  the explicit expressions of $E(-1)$ and $E'(-1)$ given in \eqref{eq: E(-1)} and \eqref{eq: Ep(-1)} that
\begin{equation} \label{Term3-final}
  \lim_{z \to -1}\left(E^{-1}(z) E'(z) \right)_{31} = 0.
\end{equation}
For the limit of the right hand side on of \eqref{eq:4term1}, we note from \eqref{eq:estR} that, as $s\to +\infty$,
\begin{equation}
  E^{-1}(z) R^{-1}(z) R'(z) E(z)  = E^{-1}(z) \left( \frac{R_1'(z)}{s^{\frac43}} + \Boh(s^{-\frac{8}{3}}) \right) E(z).
\end{equation}
In view of the explicit expressions of $E(-1)$ and $R_1'(-1)$ in \eqref{eq: E(-1)} and \eqref{Rp(-1)-expression}, it is then readily seen that
\begin{equation} \label{Term4-final}
  \lim_{z \to -1}\left( E^{-1}(z) R^{-1}(z) R'(z) E(z) \right)_{31} = \pi i   \frac{C_1 - 12 C_3 }{32 C_1} + \Boh(s^{-\frac{4}{3}}), \quad s\to +\infty,
\end{equation}
where the constants $C_i$, $i=1,3$, are given in \eqref{def:Ci}.

Finally, by \eqref{eq:derivativeinsT}, \eqref{T'formula-A-B}, \eqref{eq:term1}, \eqref{Term2-final}, \eqref{eq:3term1}, \eqref{eq:4term1}, \eqref{Term3-final} and \eqref{Term4-final}, we obtain
\begin{align}
\frac{\partial}{\partial s} F(s;\rho)&=-\frac{1}{ \pi i s} \lim_{z \to -1} \left[\left(T^{-1}(z)T'(z)\right)_{21}+\left(T^{-1}(z) T'(z)\right)_{31}\right]
\nonumber
\\
& = -\frac{1}{s} \left( \frac{1}{2}s^{\frac{8}{3}} C_1^2 + \frac{C_1 - 12 C_3 }{32 C_1} + \Boh(s^{-\frac{4}{3}}) \right)
\nonumber
\\
&=-\frac{3 s^{\frac53}}{2^{\frac83}} + \frac{\rho s}{2} - \frac{\rho^2 s^{\frac13}}{3 \cdot 2^{\frac43}} - \frac{2}{9 s} + \Boh(s^{-\frac{5}{3}}).
\end{align}
as $s \to +\infty$. Integrating the above formula gives us
\begin{equation} \label{main-F-asy-2}
 F(s;\rho)= -\frac{9 s^{\frac83}}{2^{\frac{17}3}} + \frac{\rho s^2}{4} - \frac{\rho^2 s^{{\frac43}}}{2^{{\frac{10}3}}} - \frac{2}{9} \ln s + \kappa(\rho)  + \Boh(s^{-\frac{2}{3}}),
\end{equation}
uniformly for $\rho$ in any compact subset of $\mathbb{R}$, where $\kappa(\rho)$ is the constant of integration that might be dependent on $\rho$.

To find more information about $\kappa(\rho)$, we come to $\frac{\partial}{\partial \rho} F(s;\rho)$. From \eqref{eq:derivativein-rho-X} and \eqref{asyS:coeff}, we have
\begin{equation} \label{eq:derivativeinRhoS}
  \frac{\partial}{\partial \rho} F(s;\rho) = -\frac{s^{\frac{2}{3}}}{2} \left( (  \mathsf{S}_1 )_{12} + ( \mathsf{S}_1 )_{23} \right) - D_1 s^2 + \frac{\rho^3}{54},
\end{equation}
where $\mathsf{S}_1$ and $D_1$ are given in \eqref{asyS:coeff} and \eqref{D0-def}, respectively. Recall that (see \eqref{def:StoR})
$$S(z) = R(z) N(z), \qquad z\in \mathbb{C}\setminus \{U(-1,\delta) \cup U(1,\delta) \cup  \Sigma_S\},$$
it follows from the large $z$ behaviors of $S(z)$,  $N(z)$ and $R(z)$  in \eqref{eq:asyS}, \eqref{eq:asyN} and \eqref{eq:asyR} that
\begin{equation}
  \mathsf{S}_1 = \mathsf{R}_1 + \mathsf{N}_1,
\end{equation}
where $\mathsf{R}_1$ and $\mathsf{N}_1$ are the coefficients of $1/z$ for $R(z)$ and $N(z)$ at infinity. A combination of the above formula and the expressions of $\mathsf{R}_1$ and $\mathsf{N}_1$ in \eqref{N1-coeff} and \eqref{R1-coeff} gives us
\begin{equation}
  (  \mathsf{S}_1 )_{12} + ( \mathsf{S}_1 )_{23} = \Boh(s^{-\frac{4}{3}}), \qquad  s \to +\infty.
\end{equation}
This, together with \eqref{eq:derivativeinRhoS} and $D_1$ given in \eqref{D0-def}, further implies
\begin{equation}
  \frac{\partial}{\partial \rho} F(s;\rho) = \frac{s^2}{4} - \frac{\rho s^{\frac{4}{3}}}{2^{\frac{7}{3}}} + \frac{\rho^3}{54} + \Boh(s^{-\frac{2}{3}}).
\end{equation}
Comparing this approximation with the asymptotics of $F(s;\rho)$ given in \eqref{main-F-asy-2}, it is easily seen that
\begin{equation}\label{eq:kapparho}
  \kappa'(\rho) = \frac{\rho^3}{54} \Longrightarrow \kappa(\rho) = \frac{\rho^4}{216} + C,
\end{equation}
where $C$ is an undetermined constant independent of $s$ and $\rho$. Inserting \eqref{eq:kapparho} into \eqref{main-F-asy-2} leads to our final asymptotic result \eqref{main-F-asy}.

This completes the proof of Theorem \ref{main-thm}.
\qed

\begin{appendices}
\section{Bessel parametrix} \label{Append: BP}
Define
\begin{equation}\label{Phi-B-solution}
\Phi^{(\Bes)}_{\alpha}(z)=\left\{
                             \begin{array}{ll}
                               \begin{pmatrix}
I_{\alpha}(z^{1/2}) & \frac{i}{\pi}K_{\alpha}(z^{1/2}) \\
\pi iz^{1/2}I'_{\alpha}(z^{1/2}) &
-z^{1/2}K_{\alpha}'(z^{1/2})
\end{pmatrix}, & z\in \texttt{I},\\ [4mm]
                              \begin{pmatrix}
I_{\alpha}(z^{1/2}) & \frac{i}{\pi}K_{\alpha}(z^{1/2}) \\
\pi iz^{1/2}I'_{\alpha}(z^{1/2}) &
-z^{1/2}K_{\alpha}'(z^{1/2})
\end{pmatrix}\begin{pmatrix}
                                1 & 0\\
                               -e^{\alpha \pi i} & 1
                                \end{pmatrix}, & z\in \texttt{II}, \\ [4mm]
                                \begin{pmatrix}
I_{\alpha}(z^{1/2}) & \frac{i}{\pi}K_{\alpha}(z^{1/2}) \\
\pi iz^{1/2}I'_{\alpha}(z^{1/2}) &
-z^{1/2}K_{\alpha}'(z^{1/2})
\end{pmatrix}\begin{pmatrix}
1 & 0 \\
e^{-\alpha\pi i} & 1
\end{pmatrix}, &  z\in \texttt{III},
                             \end{array}
                           \right.
\end{equation}
where  $I_\alpha(z)$ and $K_\alpha(z)$ denote the  modified Bessel functions (cf. \cite[Chapter 10]{DLMF}), the principle branch is taken for $z^{1/2}$ and the regions $\texttt{I-III}$ are illustrated in Fig. \ref{fig:jumps-Phi-B}. By \cite{KMVV},  we have that $\Phi^{(\Bes)}_{\alpha}(z)$ satisfies the RH problem below.

   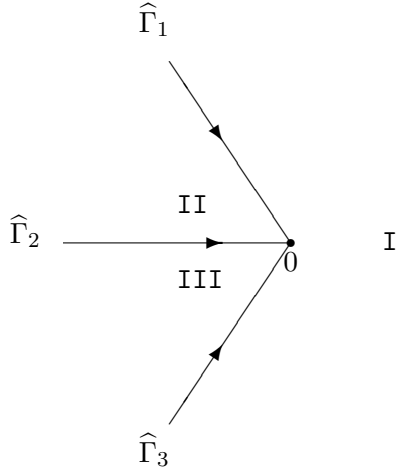
\begin{figure}[t]
\begin{center}
   \setlength{\unitlength}{1truemm}
   \begin{picture}(80,70)(-5,2)
       \put(40,40){\line(-2,-3){16}}
       \put(40,40){\line(-2,3){16}}
       \put(40,40){\line(-1,0){30}}

       \put(30,55){\thicklines\vector(2,-3){1}}
       \put(30,40){\thicklines\vector(1,0){1}}
       \put(30,25){\thicklines\vector(2,3){1}}

       \put(39,36.3){$0$}
       \put(20,11){$\widehat \Gamma_3$}
       \put(20,68){$\widehat \Gamma_1$}
       \put(3,40){$\widehat  \Gamma_2$}

       \put(52,39){$\texttt{I}$}
       \put(25,44){$\texttt{II}$}
       \put(25,34){$\texttt{III}$}

       \put(40,40){\thicklines\circle*{1}}

   \end{picture}
   \caption{The jump contours and regions for the RH problem for $\Phi^{(\Bes)}_{\alpha}$.}
   \label{fig:jumps-Phi-B}
\end{center}
\end{figure}

\subsection*{RH problem for $\Phi^{(\Bes)}_\alpha$}
\begin{description}
\item(a) $\Phi^{(\Bes)}_{\alpha}(z)$ is defined and analytic in $\mathbb{C}\setminus \{\cup^3_{j=1}\widehat \Gamma_j\cup\{0\}\}$, where the contours $\widehat \Gamma_j$, $j=1,2,3$,  are indicated  in Figure \ref{fig:jumps-Phi-B}.

\item(b) $\Phi^{(\Bes)}_{\alpha}(z)$ satisfies the jump condition
\begin{equation} \label{Bessel-jump}
 \Phi^{(\Bes)}_{\alpha,+}(z)=\Phi^{(\Bes)}_{\alpha,-}(z)
 \left\{
 \begin{array}{ll}
   \begin{pmatrix}
                                1 & 0\\
                               e^{\alpha \pi i} & 1
                                \end{pmatrix},  &  \qquad z \in \widehat \Gamma_1, \\
   \begin{pmatrix}
                                0 & 1\\
                               -1 & 0
                                \end{pmatrix},  &  \qquad z \in \widehat \Gamma_2, \\
   \begin{pmatrix}
                                 1 & 0 \\
                                 e^{-\alpha\pi i} & 1 \\
                                \end{pmatrix}, &   \qquad z \in \widehat \Gamma_3.
 \end{array}  \right .
 \end{equation}

\item(c) $\Phi^{(\Bes)}_{\alpha}(z)$ satisfies the following asymptotic behavior at infinity:
\begin{multline}\label{eq:Besl-infty}
 \Phi^{(\Bes)}_{\alpha}(z)=
 \frac{( \pi^2 z )^{-\frac{1}{4} \sigma_3}}{\sqrt{2}}
 \begin{pmatrix}
 1 & i
 \\
 i & 1
 \end{pmatrix}
 \\
 \times  \left( I + \frac{1}{8 z^{1/2}}  \begin{pmatrix}
   -1-4\alpha^2 & -2i \\
   -2i & 1+4\alpha^2
 \end{pmatrix} +  \Boh\left(\frac{1}{z}\right)
 \right)e^{z^{1/2}\sigma_3},\quad z\to \infty.
   \end{multline}

\item(d) $\Phi^{(\Bes)}_{\alpha}(z)$ satisfies the  following asymptotic behaviors near the origin:
\newline
    If $\alpha<0$,
    \begin{equation}
\Phi^{(\Bes)}_{\alpha}(z)=
\Boh \begin{pmatrix}
|z|^{\alpha/2} & |z|^{\alpha/2}
\\
|z|^{\alpha/2} & |z|^{\alpha/2}
\end{pmatrix}, \qquad \textrm{as $z \to 0$}.
\end{equation}
If $\alpha=0$,
    \begin{equation}
\Phi^{(\Bes)}_{\alpha}(z)=
\Boh \begin{pmatrix}
\ln|z| & \ln|z|
\\
\ln|z| & \ln|z|
\end{pmatrix}, \qquad \textrm{as $z \to 0$}.
\end{equation}
If $\alpha>0$,
     \begin{equation}
\Phi^{(\Bes)}_{\alpha}(z)= \left\{
                             \begin{array}{ll}
                                \Boh\begin{pmatrix}
|z|^{\alpha/2} & |z|^{-\alpha/2}
\\
|z|^{\alpha/2} & |z|^{-\alpha/2}
\end{pmatrix}, & \hbox{as $z \to 0$ and $z\in \texttt{I}$,}
\\
 \Boh \begin{pmatrix}
|z|^{-\alpha/2} & |z|^{-\alpha/2}
\\
|z|^{-\alpha/2} & |z|^{-\alpha/2}
\end{pmatrix}, & \hbox{as $z \to 0$ and $z\in \texttt{II}\cup \texttt{III} $.}
                             \end{array}
                           \right.
\end{equation}
\end{description}

\end{appendices}

\section*{Acknowledgements}
Dan Dai was partially supported by a grant from the City University of Hong Kong (Project No. 7005252), and grants from the Research Grants Council of the Hong Kong Special Administrative Region, China (Project No. CityU 11303016, CityU 11300520). Shuai-Xia Xu was partially supported by National Natural Science Foundation of China under grant numbers 11971492, 11571376 and 11201493. Lun Zhang was partially supported by National Natural Science Foundation of China under grant numbers 11822104 and 11501120, by The Program for Professor of Special Appointment (Eastern Scholar) at Shanghai Institutions of Higher Learning, and by Grant EZH1411513 from Fudan University. He also thanks Marco Bertola for helpful discussions related to this work.


\end{document}